\documentclass[12pt,journal,onecolumn,draftclsnofoot]{IEEEtran}
\usepackage[T1]{fontenc}
\usepackage[utf8]{inputenc}
\usepackage{authblk} 
\usepackage{amsthm} 
\usepackage{mathbbol}
\usepackage{cite}
\include{psfig}
\usepackage{psfig}
\usepackage{psfrag}
\usepackage{amssymb}
\usepackage{picins}
\usepackage[demo]{graphicx}
\graphicspath{ {images/} }
\usepackage{wrapfig}
\usepackage{graphicx}
\usepackage{tikz}
\usepackage{subfigure}
\usepackage{enumitem}
\usepackage{tikz} 
\usetikzlibrary{calc}
\usepackage[tbtags]{amsmath}
\usepackage{epsfig} 
\usepackage{epstopdf}
\usepackage{stfloats} 
\usepackage{float}
\usepackage{array}
\usepackage{color}
\usepackage{pgfplots}
\usepackage{textcomp} 
\usepackage{gensymb}
\usepackage{multirow}
\usepackage{url}
\usepackage{algorithm}
\usepackage{algorithmic}
\pgfplotsset{compat=newest}
\pgfplotsset{plot coordinates/math parser=false} 
\usepackage{stmaryrd} 
\newtheorem{theorem}{Theorem} 

\IEEEoverridecommandlockouts 
\newlength\tindent
\renewcommand{\indent}{\hspace*{\tindent}} 

\begin{document} 
\title{\huge{A Low-Complexity Framework for Joint User Pairing and Power Control for Cooperative NOMA in 5G and Beyond Cellular Networks}}
\author{\IEEEauthorblockN{Phuc Dinh$^*$},\thanks{Phuc Dinh, Mohamed Amine Arfaoui and Chadi Assi are with Concordia Institute for Information Systems Engineering (CIISE), Concordia University, Montreal, Canada, e-mails:$\{\text{p\_dinh}$@encs, $\text{m\_arfaou}$@encs, assi@ciise$\}$.concordia.ca. Sanaa Sharafeddine is with the the School of Arts and Sciences (SAS), Beirut, Lebanon, email: sanaa.sharafeddine@lau.edu.lb. Ali Ghrayeb is with the Electrical and Computer Engineering (ECE) department, Texas A \& M University at Qatar, Doha, Qatar, e-mail: ali.ghrayeb@qatar.tamu.edu. Part of this work will be presented at IEEE Globecom 2019 \cite{Phuc2019joint}. The statements made herein are solely the responsibility of the authors.} \IEEEauthorblockN{Mohamed Amine Arfaoui}, \IEEEauthorblockN{Sanaa Sharafeddine}, \IEEEauthorblockN{Chadi Assi} and \IEEEauthorblockN{Ali Ghrayeb}} 
\maketitle 
\thispagestyle{plain}
\begin{abstract}
This paper investigates the performance of cooperative non-orthogonal multiple access (C-NOMA) in a cellular downlink system. The system model consists of a base station (BS) serving multiple users, where users with good channel quality can assist the transmissions between the BS and users with poor channel quality through either half-duplex (HD) or full-duplex (FD) device-to-device (D2D) communications. We formulate and solve a novel optimization problem that jointly determines the optimal D2D user pairing and the optimal power control scheme, where the objective is maximizing the achievable sum rate of the whole system while guaranteeing a certain quality of service (QoS) for all users. The formulated problem is a mixed-integer non-linear program (MINLP) which is generally NP-hard. To overcome this issue, we reconstruct the original problem into a bi-level optimization problem that can be decomposed into two sub-problems to be solved independently. The outer problem is a linear assignment problem which can be efficiently handled by the well-known Hungarian method. The inner problem is still a non-convex optimization problem for which finding the optimal solution is challenging. However, we derive the optimal power control policies for both the HD and the FD schemes in closed-form expressions, which makes the computational complexity of the inner problems polynomial for every possible pairing configurations. These findings solve ultimately the original MILNP in a timely manner that makes it suitable for real-time and low latency applications. Our simulation results show that the proposed framework outperforms a variety of proposed schemes in the literature and that it can obtain the optimal pairing and power control policies for a network with 100 users in a negligible computational time. 
\end{abstract}
\section{Introduction} \label{sec:introduction}
\subsection{Motivation}
\indent The total data traffic is expected to become about 49 exabytes per month by 2021, while in 2016, it was approximately 7.24 exabytes per month \cite{Intro1}. With this dramatic increase, fifth-generation (5G) networks and beyond must urgently provide high data rates, seamless connectivity and ultra-low latency communications \cite{Intro2,Intro3,Intro4}. In addition, with the emergence of the Internet-of-things (IoT) networks, the number of connected devices to the internet is increasing exponentially \cite{Intro5,Intro6}. This fact implies not only a significant increase in data traffic, but also the emergence of some IoT services with crucial requirements. Such requirements include higher data rates, higher connection density and ultra reliable low latency communication (URLLC). Hence, traditional radio-frequency (RF) networks, which are already crowded, are unable to satisfy these high demands \cite{Intro7}. Network densification \cite{Intro8,Intro9} has been proposed as a solution to increase the capacity and coverage of 5G networks. However, with the continuous dramatic growth in data traffic, researchers from both industry and academia are trying to explore new network architectures, new transmission techniques and new communication technologies to meet these demands. Among the new communication technologies that have been proposed as auspicious solutions for 5G and beyond are non-orthogonal multiple access (NOMA), device-to-device (D2D) communication, half duplex (HD) and full-duplex (FD) communications. \\ 
\indent NOMA is capable of supporting more users than the number of available orthogonal resources \cite{islam2017power}, thereby leading to higher spectral efficiency and user fairness when compared to standard orthogonal multiple access (OMA) techniques.\footnote{In this paper, the term "NOMA" is restricted to power-domain NOMA as distinct from its code-domain NOMA counterpart.} The principle of NOMA leverages the concept of superposition coding (SPC) at the transmitter, to multiplex users in power-domain, and successive interference cancellation (SIC) at the receiver \cite{6861434}. However, as a standalone technique, NOMA still cannot adequately fulfil the demanding specifications of 5G networks and beyond. In fact, an inherent limitation of NOMA lies in the requirement that the allocated power to a user with poor channel conditions needs to be high for successful decoding of the superimposed signal. This requirement generally reduces the spectral (and power) efficiency since the poor channel will absorb a large portion of the available power budget. Thus, recent research trends aim to modularize and integrate NOMA with other advanced and 5G-enabled techniques, such as multiple-input-multiple-output (MIMO) \cite{liu2016capacity}, hybrid automatic repeat request (HARQ)\cite{li2015investigation}, and device-to-device (D2D) communications \cite{zhao2016noma}. \\ 
\indent A recent appealing extension for NOMA is cooperative NOMA (C-NOMA), taking advantage of desirable attributes of NOMA, FD/HD, and D2D communications. For a practical scenario, each user with a high channel gain, which is referenced as strong user, can act as relay to assist communications between the transmitter and a user with a poor channel gain, which is referenced as weak user. Each weak user can then combine both signals coming from the transmitter and from the associated strong user. The novelty of C-NOMA revolves around adding the D2D communication factor as a new degree of diversity to the direct downlink transmission. Intuitively speaking, each weak user can receive its message either from the BS or from another user, depending on which one is more favorable. 
\subsection{Literature Review}
The performance gain of NOMA comes at the expense of design complexity. Therefore, extensive performance analysis and efficient performance optimization scheme under a wide range of network scenarios have been studied in  \cite{Ding:NOMA:Application,islam2017powerNOMA,PA:NOMA1,PA:NOMA2,PA:NOMA3, PA:NOMA4, Islam:ResourceAllocation:NOMA,NOMAPairing:MatchingAlgorithm,Letter:NOMA:Pairing,Ding:Impact:Pairing,Ding:NOMA:ModeSelection}. 
In \cite{Ding:NOMA:Application}, the authors presented the basis of NOMA and the relations of NOMA to a few network paradigms such as MIMO or cognitive radio (CV). In \cite{islam2017powerNOMA},  Islam \textit{et al} provided a comprehensive survey on the recent progress of NOMA in terms of capacity analysis, power allocation, pairing, and user fairness. As a major aspect of NOMA literature, resource allocation for NOMA-based networks are investigated in \cite{PA:NOMA1,PA:NOMA2,PA:NOMA3, PA:NOMA4, Islam:ResourceAllocation:NOMA}. To mitigate the impact of error propagation due to imperfect SIC and to reduce the overall complexity, pairing (and clustering) is essential for networks that have large number of users. Motivated by this, optimal pairing policy for NOMA has been studied in  \cite{NOMAPairing:MatchingAlgorithm,Letter:NOMA:Pairing,Ding:Impact:Pairing}. User pairing does not only offer enhanced scalability and modularity for system optimization but also entails potential usage of D2D communications among users to further enhance system throughput, energy efficiency and fairness. Optimal pairing scheme for sum-rate maximization for multi-user NOMA networks has been proposed in \cite{Zhu:OptimalPairing}. Thus, the literature on NOMA is considered rich, yet NOMA still cannot adequately fulfil the demanding specifications of 5G networks and beyond in terms of high data rates and massive connectivity. \\ 
\indent By combining NOMA with HD/FD and D2D communications, C-NOMA takes the advantage of the SIC process combined with decode-and-forward procedure at users with high channel conditions to increase the reception diversity at the users that experience severe channel fading, and hence increasing the total throughput \cite{7117391}. As opposed to NOMA, the research on C-NOMA is still far from being mature. One observation is that both the pairing policy and the power allocation scheme of C-NOMA systems are not a direct inheritance of those of conventional NOMA. One of the key challenges of investigating C-NOMA lies in the complicated achievable rate expressions which capture the  characteristics of NOMA as well as the HD/FD decode-and-forward procedures. Specifically, the achievable rate of a given user within a C-NOMA system is a non-concave and non-differentiable function, which hinders the direct application of derivative-based numerical optimization frameworks \cite{BoydS:98:LAA}. In \cite{C-NOMA:2users1,C-NOMA:2users2,Liu:C-NOMA:2users,C-NOMA:2users3,8417519}, the performance of C-NOMA, measured in terms of outage probability, error performance and capacity, were investigated. However, these works focused mainly on two-user setting for simplicity and tractability purposes. In addition, the adopted performance metric is the max-min achievable rate, which is known to be bandwidth-inefficient since the majority of resource is allocated to the user with poor channel gain to maintain fairness. A limited amount of works in the literature have formally studied the performance of multi-user C-NOMA \cite{Guo:MultiuserCNOMA, Zhou:multiuserCNOMA1,Liu:MultiuserCNOMA2:SWIPT, Alouini:VLC:C-NOMA}. In \cite{Guo:MultiuserCNOMA,Zhou:multiuserCNOMA1}, the authors studied the performance of C-NOMA in terms of outage performance. However, these works have either employed random or heuristic distance-based pairing, which limits the performance gains of C-NOMA. Additionally, in \cite{Zhou:multiuserCNOMA1}, simultaneous wireless information and power transfer was integrated with a C-NOMA system to reduce the energy consumption at the users. However, the pairing policy and power allocation problem were not investigated. Recently, Obeed \textit{et al}, \cite{Alouini:VLC:C-NOMA} studied the joint problem of user pairing and power allocation with HD links and a fixed relaying power for hybrid radio-frequency (RF) and visible light communication (VLC) systems. However, assuming a fixed relaying power is not a practical assumption for realistic scenarios. In fact, assuming that the strong user uses the maximum allowed relaying power is technically unappealing, since power at user devices is in general small, the energy efficient of C-NOMA schemes should be considered carefully instead of assuming fixed powers at the relaying devices. 
\subsection{Contributions}
We investigate in this paper the performance of C-NOMA in a cellular downlink system that consists of a BS that wants to serve multiple users within a region of service. In this system, users that have the capability of either HD or FD communications can assist the transmissions between the BS and users with poor channel quality through D2D communications. We formulate and solve a novel optimization problem that jointly determines the optimal D2D user pairing policy, the optimal power control scheme, and the best D2D communication mode, i.e., HD or FD, where the objective is maximizing the achievable sum rate of the whole system while guaranteeing a certain quality of service (QoS) for all users. \\ 
\indent In light of the above background, and to the best of our knowledge, this problem has not been investigated in the literature. The proposed formulation is modular and can be applied for all possible relaying schemes, such as HD, FD or hybrid HD/FD. However, the formulated problem ends up to be a mixed-integer non-convex program, which involves not only binary decision variables but also non-smooth utility functions. Thus, derivative-based convexification methods can not be directly applicable \cite{KTJ:16:TSP}.\footnote{Note that most previous works considering this class of problems typically assume overprovisioning of resources, i.e, high power and low minimum-rate requirement, to guarantee the feasibility to apply proposed algorithms.} Alternatively, we reformulate the problem into a bi-level optimization problem, consisting of one outer problem and one inner problem.\footnote{It is important to clarify that the proposed bilevel optimization method is fundamentally different from BCD method proposed by \cite{Alouini:VLC:C-NOMA} although both involve decomposition of the original problem into subproblems}. The outer problem is a classic linear assignment problem. Thus, standard matching algorithms such as the Hungarian method can be applied. For the inner problem, we derive the feasibility conditions of the formulated problem as relations between the power resources, the QoS requirements, and the channel state information (CSIs). Then, we derive the closed-form solutions of the power control scheme for both the HD and FD cases. These closed-form solutions, with computational complexity $\mathcal{O}(1)$, facilitates the application of the Hungarian method, and thus, the whole problem can be solved in a polynomial time. \\
\indent The rest of the paper is organized as follows. Section \ref{sec:systemmodel} presents the system model. Section \ref{sec:Achievable Rates Analysis} presents the achievable rates analysis. Section \ref{sec:Problem Formulation and Solution Approach} presents the optimal pairing policy and power control scheme. Sections \ref{sec:Simulation Results} and \ref{Sec:Conclusion} presents the simulation results and the conclusion, respectively. 
\section{System Model}
\label{sec:systemmodel}
\subsection{Network Model}
\begin{figure}[t]
\centering 
\includegraphics[width = 0.6\linewidth, height = 0.25\textwidth]{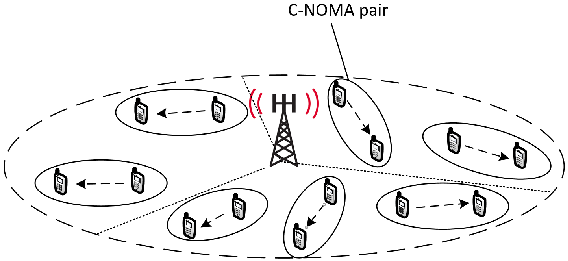}
\caption{Proposed downlink C-NOMA for a cellular network. The network is divided into $M = 3$ sectors that are served via orthogonal channels.}
\label{fig:Spatial}
\end{figure}
We consider a standard single-cell cellular system that consists of a BS equipped with $M$ antennas and serving $M$ disjoint sectors, where each sector is served with one antenna as shown in Fig.~\ref{fig:Spatial}. In this paper, we investigate the performance of one typical sector within which we assume that there exist $2K$ spatially dispersed active users (AUs). It is important to highlight here that considering an odd number of AUs does not affect the generality of the system model, since if the number of AUs is even, we still can adopt the same network model by adding an extra virtual AU which has zero channel gain. The AUs can be divided based on their channel gains into two disjoint sets of users, which we denote by $\mathcal{S}$ and $\mathcal{W}$. The set $\mathcal{S}$ contains the AUs that have high channel gain and such users are referred to as "strong AUs". On the other hand, the set $\mathcal{W}$ contains the users with low channel gains and these users are referred to as "weak AUs". According to C-NOMA principle \cite{7117391}, each weak AU from $\mathcal{W}$ is paired with exactly one strong AU from $\mathcal{S}$ and the resulting distinct pairs of (strong AU, weak AU) are served simultaneously over orthogonal and equally divided frequency subchannels in order to cancel the inter-pair interference. Moreover, the AUs within each pair are served using NOMA principle, where additionally, the strong AU can assist the communication between the BS and the weak AU. Obviously, the system performance depends on the pairing configurations and the communication links within each pair, which is the focus of this paper. In the following subsection, we investigate the transmission model within each pair of AUs.
\subsection{Transmission Model}
\begin{figure}[t]
\centering
\includegraphics[width=0.5\linewidth, height = 0.4\textwidth]{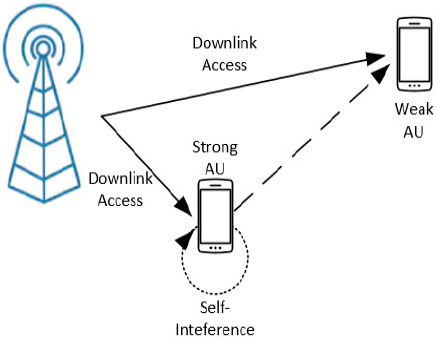}
\caption{Illustration of the transmission phases within a FD C-NOMA pair.}
\label{fig:Transmission}
\end{figure}
According to C-NOMA principle, and as presented in Fig.~\ref{fig:Transmission}, the transmission model within each pair of AUs consists of two phases that are detailed as follows. 
\begin{enumerate}
    \item The NOMA downlink transmission phase: the BS applies superposition coding (SC) on the messages intended to the strong and weak AUs and transmit the superimposed message to both of them \cite{6204010}. Then, following NOMA principle, the weak AU treats the interference from the strong AUs as noise and decode its own message \cite{6204010}.
    \item The D2D cooperative relaying phase: the strong AU first applies SIC to decode the intended message for weak AU. Second, it subtracts the decoded message of the weak AU from its own reception. Then, it decodes its own message from its resulting interference-free reception \cite{6861434}. Finally, it relays this decoded message for the weak AU through a D2D channel, therefore enhancing the signal reception diversity \cite{Caire:DecodeAndForward}.
\end{enumerate}
For the case of HD communication, the two transmission phases occur on consecutive resource blocks. However, for the case of FD communication, they occur on the same resource block, which comes with the cost of inducing SI at the strong AU (the dotted line in Fig. \ref{fig:Transmission}).
\section{Achievable Rate Analysis}
\label{sec:Achievable Rates Analysis}
In this section, the analysis of the achievable rates for a downlink C-NOMA system for both HD and FD cases are presented. Consider a pair of AUs $(m,n) \in \llbracket 1,K \rrbracket^2$, where $m$ and $n$ denote the indices of the strong and the weak AU, respectively, within the pair. For both cases of HD and FD relaying and at each channel use, the resulting signal at the BS after applying SC is expressed as
\begin{equation}
s =\sqrt{ \alpha_{m,n} P_{\rm BS} } s_{n} + \sqrt{\left(1-\alpha_{m,n}\right)  P_{\rm BS} } s_{m},
\label{eq:transmitsignal}
\end{equation}
where $s_m$ and $s_n$ represent the intended messages of the strong AU and the weak AU, respectively, such that $\mathbb{E}(|s_n|^{2})=\mathbb{E}(|s_m|^{2})=1$, $P_{\rm BS}$ represents the total average power available at the BS and $\alpha_{m,n} \in [0,1]$ is the power allocation factor, i.e., $\left(1-\alpha_{m,n} \right)  P_{\rm BS}$ and $\alpha_{m,n} P_{\rm BS}$ represent the transmit powers allocated to the strong AU and the weak AU, respectively. 
\subsection{HD C-NOMA}
For the case of HD C-NOMA, the received signal at the strong AU at each channel use is given by
\begin{equation}
\begin{split}
y_m = h_m \left(\sqrt{\alpha_{m,n} P_{\rm BS}}s_n+ \sqrt{\left(1-\alpha_{m,n}\right)P_{\rm BS}}s_m \right) +\omega_m,
\end{split}
\label{eq:strongreceivedsignal-half}
\end{equation}
where $h_m$ is the channel gain from the BS to the strong AU and $\omega_m$ is a zero-mean unit-variance additive Gaussian noise (AWGN). At the strong AU, SIC is applied to decode the message $s_n$ of the weak AU. Therefore, the achievable rate of the strong AU to decode the message intended for the weak AU can be expressed as
\begin{equation}
R^{\rm H}_{m,n} = \frac{1}{2}\log_2\left(1+ \frac{\alpha_{m,n} P_{\rm BS}\gamma_m}{(1-\alpha_{m,n})P_{\rm BS}\gamma_m+1}\right),
\label{Rate-Hsw}
\end{equation}
where $\gamma_m = h_m^2$. After subtracting the message $s_n$ from its reception, the strong AU can decode its own message with an achievable rate that is given by
\begin{equation}
R^{\rm H}_{m} = \frac{1}{2}\log_2 \left(1+(1-\alpha_{m,n}) P_{\rm BS}\gamma_m\right)\label{Rate-Hs}.
\end{equation}
On the other hand, the received signal at the weak AU from the BS in the direct transmission phase is expressed as
\begin{equation}
\begin{split}
y_n = h_n\left(\sqrt{\alpha_{m,n} P_{\rm BS}}s_n+\sqrt{(1-\alpha_{m,n})P_{\rm BS}}s_m\right)+ \omega_n,
\end{split}
\end{equation}
where $h_n$ is the channel gain from the BS to the weak AU and $\omega_n$ is a zero-mean unit-variance AWGN. In the cooperative relaying phase, the strong AU forwards the message $s_n$ to the weak AU. Thus, the received signal at the weak AU in this phase is given by
\begin{equation}
y_n = h^{\rm d}_{m,n}\sqrt{P^{\rm d}_{m,n}}s_n + \omega_n,
\end{equation}
where $h^{\rm d}_{m,n}$ is the channel gain from the strong AU to the weak AU and $P^{\rm d}_{m,n}$ is the D2D transmit relaying power. Since the weak AU receives duplicate messages from both the BS and the strong AU, repetition decoding (RD) can be applied to decode the replicated information \cite{Bossert:ChannelCodingBook}. In other words, the effective signal-to-interference-plus-noise-ratio (SINR) at the weak AU is the summation of the SINRs of the transmission links from the BS to the weak AU and from the strong AU to the weak AU. In this case, the achievable rate of the weak AU is expressed as
\begin{equation}
R^{\rm RD}_{n} = \frac{1}{2}\log\left(1+P^{\rm d}_{m,n}\gamma^{\rm d}_{m,n}+\frac{\alpha_{m,n} P_{\rm BS} \gamma_{n}}{(1-\alpha_{m,n})P_{\rm BS}\gamma_{n}+1}\right),
\label{eq:RateMRCH}
\end{equation}
where $\gamma^{\rm d}_{m,n} = \left(h^{\rm d}_{m,n}\right)^2$ and $\gamma_{n} = h_{n}^2$. Finally, in line with the results of \cite{Laneman:CooperativeDiveristy}, the resulting achievable rate of the weak AU from the cooperative diversity of the BS and the strong AU is expressed as
\begin{equation}
\begin{split}
R^{\rm H}_{n} = \min \left(R^{\rm RD}_{n},\,\, R^{\rm H}_{m,n} \right).
\end{split}
\label{eq:RateHwfinal}
\end{equation}
\subsection{FD C-NOMA}
For the case of FD C-NOMA, the received signal at the strong AU is given by
\begin{equation}
y_m = \,\,h_m \left(\sqrt{\alpha_{m,n} P_{\rm BS}}s_n+\sqrt{\left(1-\alpha_{m,n}\right)P_{\rm BS}}s_m\right) + h^{\rm SI}_{m}\sqrt{P^{\rm d}_{m,n}}s_n+\omega_m,
\label{eq:strongreceivedsignal}
\end{equation}
where $h^{\rm SI}_m$ represents the SI channel gain at the strong AU  (the dotted line in Fig. \ref{fig:Transmission}). Thus, the achievable rate of the strong AU to decode the message of the weak AU is expressed as
\begin{equation}
R^{\rm F}_{m,n} = \log_2\left(1+ \frac{\alpha_{m,n} P_{\rm BS}\gamma_m}{(1-\alpha_{m,n})P_{\rm BS}\gamma_m+P^{\rm d}_{m,n}\gamma^{\rm SI}_{m}+1}\right),
\label{eq:Ratem-1}
\end{equation}
where $\gamma^{\rm SI}_{m} = \left(h^{\rm SI}_{m}\right)^2$. Then, after successfully decoding and canceling the message $s_{n}$ of the weak AU, the strong AU decodes his own message $s_{m}$. Therefore, the achievable rate of the strong user to decode its own message is expressed as 
\begin{equation}
R^{\rm F}_m = \log_2 \left(1+\frac{(1-\alpha_{m,n}) P_{\rm BS}\gamma_m}{P^{\rm d}_{m,n}\gamma^{\rm SI}_{m}+1} \right).
\label{eq:Ratem}
\end{equation}
Afterwards, the strong AU forwards the message intended for the weak AU. Assuming that the processing delay caused by the SIC process at the strong AU is small, the weak AU receives the message $s_n$ from the BS and from the strong user at approximately the same channel use \cite{8094966}. Therefore, the total received signal at the weak AU is given by
\begin{equation}
y_{n} = h_n \left (\sqrt{\alpha_{m,n} P_{\rm BS}}s_n+\sqrt{(1-\alpha_{m,n})P_{\rm BS}}s_m \right)+h^{\rm d}_{m,n}\sqrt{P^{\rm d}_{m,n}}s_{\rm w}+ \omega_{n},
\label{eq:weakreceivedsignal}  
\end{equation}
Following \cite{1310306}, we assume that the weak AU can successfully co-phase and combine the signals from the BS and the strong AU by a proper diversity combining technique such as the maximum ratio combining (MRC). In this case, the effective SINR of the weak AU is the summation of the SINRs resulting from decoding its received message from the BS and the strong AU. Consequently, the achievable rate of the weak AU when applying MRC can be written as 
\begin{equation}
    R^{\rm MRC}_{n} = \log_2\left(1+ P^{\rm d}_{m,n}\gamma^{\rm d}_{m,n}+\frac{\alpha_{m,n} P_{\rm BS}\gamma_{n}}{(1-\alpha_{m,n}) P_{\rm BS}\gamma_{n}+1}\right),
\label{eq:RateMRc}
\end{equation}
Based on the above analysis and the results of \cite{4205048}, the resulting achievable rate of the weak AU is given by
\begin{equation}
R^{\rm F}_{n} =\min \left(R_{m,n}^{\rm F}, \,\,R^{\rm MRC}_{n} \right).
\label{eq:rate-weak}
\end{equation}
\section{Optimal Pairing Policy and Power Control Scheme}
\label{sec:Problem Formulation and Solution Approach}
\subsection{Problem Statement}
The objective of this paper is maximizing the sum rate of a downlink C-NOMA cellular sector consisting of a set of $2K$ AUs, while guaranteeing a certain QoS for each user. The maximization is performed with respect to the power allocation at the BS and the D2D links, the user pairing policy and the D2D transmission mode, i.e., HD or FD. This objective is expressed in a formal optimization problem in the following subsection.
\subsection{Problem Formulation}
Assuming that each sector consists of $2K$ AUs, the number of possible pairing configurations is $(2K-1)!! = (2K-1) \times (2K-3) \times (2K-5) \times \cdots \times 1$. Therefore, exhaustively trying all possible configurations is not practically appealing. To overcome this issue, and in order to provide a scalable solution, we instead reduce the pairing policy into a linear assignment problem, where each AU from the group of strong AUs $\mathcal{S}$ is paired with one AU from the group of weak AUs $\mathcal{W}$. Consequently, the joint pairing and power control for the sum-rate maximization problem of the overall downlink C-NOMA system can be given by the following optimization problem.
\begin{subequations}
\label{prob:P1}
\begin{align}
\mathcal{P}: \quad&R^*=\max_{\mathbf{B},\boldsymbol{\alpha},\mathbf{P}^{\rm d}}\sum_{m=1}^{K}\sum_{n=1}^{K}b_{m,n} R_{m,n} \left(\alpha_{m,n},P^{\rm d}_{m,n} \right),\label{eq:obj} \\
&\text{s.t.}\,\,\,\,  0\leq\alpha_{m,n}\leq b_{m,n},\,\, \forall \, \, m,n\in \llbracket 1,K \rrbracket,\label{constraint:budget1}\\
&\quad \,\,\,\,\,  0\leq P^{\rm d}_{m,n}\leq b_{m,n}\bar{P}^{\rm d},\,\, \forall \, \, m,n\in \llbracket 1,K \rrbracket,\label{constraint:budget2}\\
&\quad \,\,\,\,\,  b_{m,n}\in\{0,1\},\,\, \forall \, \, m,n\in \llbracket 1,K \rrbracket,\label{constraint:binary}\\
&\quad \,\,\,\,\,  \sum_{m=1}^{K}b_{m,n}=1,\,\, \forall \, \, m,n\in \llbracket 1,K \rrbracket,\label{constraint:summ}\\
&\quad \,\,\,\,\,  \sum_{n=1}^{K}b_{m,n}=1,\,\,\forall \, \, m,n\in \llbracket 1,K \rrbracket,\label{constraint:sumn}\\
&\quad \,\,\,\,\, b_{m,n}R_{m}(\alpha_{m,n},P^{\rm d}_{m,n})\geq R_{\rm th} ,\,\,\forall \, \, m,n\in \llbracket 1,K \rrbracket,\label{minimumRateFull1}\\
&\quad \,\,\,\,\, b_{m,n}R_{n}(\alpha_{m,n},P^{\rm d}_{m,n})\geq R_{\rm th} ,\,\,\forall \, \, m,n\in \llbracket 1,K \rrbracket,\label{minimumRateFull2}
\end{align}
\end{subequations}
where $R_{m,n}$ represents the sum rate of the pair $(m,n)$ that is expressed as 
\begin{equation}
    R_{m,n} = R_{m} + R_{n},
\end{equation}
such that the rate functions $\left(R_{m},R_{n} \right) = \left(R^{\rm H}_{m},R^{\rm H}_{n} \right)$ for the case of HD C-NOMA and $\left(R_{m},R_{n} \right) = \left(R^{\rm F}_{m},R^{\rm F}_{n} \right)$ for the case of FD C-NOMA, $b_{m,n}$ is the pairing decision variables, where $b_{m,n}= 1$ indicates that AU $m$ in set $\mathcal{S}$ is paired with AU $n$ in set $\mathcal{W}$ and $b_{m,n}= 0$, otherwise, $\mathbf{B}=\left\{b_{m,n} \left| m,n\in\llbracket 1,K \rrbracket \right.\right\}$, $\alpha_{m,n}$ is the "potential power allocation coefficient" of the pair $(m,n)$, $\boldsymbol{\alpha} =\left\{\alpha_{m,n}\left| m,n\in\llbracket 1,K \rrbracket \right.\right\}$, $P^{\rm d}_{m,n}$ denotes the "potential D2D transmit power" within the pair $(m,n)$, $\mathbf{P}^{\rm d} =\left\{P^{\rm d}_{m,n},\left| m,n\in\llbracket 1,K \rrbracket \right.\right\}$ and $\bar{P}^{\rm d}$ denote the D2D power budget at each AU's device.  Constraints \eqref{constraint:budget1} and \eqref{constraint:budget2} ensure that $\alpha^{\rm (m,n)} \in [0,1]$ and $P^{\rm d}_{m,n} \in \left[0, \bar{P}^{\rm d}\right]$, respectively, when AU $m$ is paired with AU $n$. Constraints \eqref{constraint:summ} and \eqref{constraint:sumn} ensure that each AU from each group can be paired with only one AU from the other group. Constraints \eqref{minimumRateFull1}-\eqref{minimumRateFull2} ensure that the paired AUs have each an achievable rate greater than a minimum achievable rate $R_{\rm th}$, whichs guarantee the QoS constraint. \\
\indent Problem $\mathcal{P}$ is a mixed-integer non-linear program (MINLP), which is generally NP-hard. Most of the previous literature resorts to iterative numerical methods such derivative-based methods or block coordinate descent (BCD) and/or off-the-shelf optimization solvers to solve this class of problems. However, BCD, which was previously adopted in \cite{Alouini:VLC:C-NOMA} is known to have poor convergence properties as shown in \cite{Tseng:Convergence:BCD} and the derivative-based methods, such as successive convex approximation (SCA), cannot be directly applied here due to the binary pairing decision variable and the non-smooth rate functions. To overcome this issue, we solve problem $\mathcal{P}$ using the concept of bilevel optimization \cite{bard2013practical}.
\subsection{Bi-level Optimization}
In problem $\mathcal{P}$, it can be observed that the rate function is only a function of power control policy and is independent upon the pairing variable $\mathbf{B}$. Precisely, let $\left\{\left. \left(b_{m,n}^*,\alpha^{*}_{m,n},P^{{\rm d}^*}_{m,n}\right) \right| m,n \in \llbracket 1,K \rrbracket \right\}$ denotes the set of optimal user pairing policy and power control scheme, which are the solutions of problem $\mathcal{P}$. For all $m,n \in \llbracket 1,K \rrbracket$ if $b_{m,n}^* = 0$, then $\left(\alpha^{*}_{m,n},P^{{\rm d}^*}_{m,n}\right) = (0,0)$. However, if $b_{m,n}^* = 1$, then $\left(\alpha^{*}_{m,n},P^{{\rm d}^*}_{m,n}\right)$ should be the optimal solutions of the power control scheme of the pair of users $(m,n)$. In other words, For all $m,n \in \llbracket 1,K \rrbracket$, if assume that the the users $(m,n)$ are paired together and that we can obtain their optimal power control scheme $\left(\alpha^{*}_{m,n},P^{{\rm d}^*}_{m,n}\right)$, problem $\mathcal{P}$ becomes a linear assignment problem and it remains to determine the optimal pairing policy $\left(b_{m,n}^* \right)_{1\leq m,n \leq K}$. Hence, since we aim to find the optimal power control that maximizes the total achievable sum rate of the C-NOMA system, we can reduce the feasible set of power control for problem $\mathcal{P}$ to the set of power control that maximizes the sum rate within a pair of users. Consequently, based on this observation, we rewrite the original problem in to a bi-level optimization \cite{Colson2007AnOO} problem as follows.
\begin{subequations}
\label{prob:P2}
\begin{align}
\mathcal{P}_{\rm outer}:\,\,&R^*=\max_{\mathbf{B}}\sum_{n=1}^{K}\sum_{m=1}^{K}b_{m,n}R_{m,n}(\alpha^{*}_{m,n},P^{{\rm d}^*}_{m,n}),\label{eq:P2:obj} \\
&\text{s.t.}\quad \eqref{constraint:binary}-\eqref{constraint:sumn}
\end{align}
\end{subequations}
where $\alpha^{*}_{m,n}$ and $P^{{\rm d}^*}_{m,n}$ are parameters obtained by solving the following problem
\begin{subequations}
\label{prob:P2.1}
\begin{align}
\mathcal{P}_{\rm inner}:\quad &R_{m,n}^*=\max_{\alpha_{m,n},P^{\rm d}_{m,n}} R_{m,n}(\alpha_{m,n},P^{\rm d}_{m,n}),\label{eq:P2.1:obj} \\
&\text{s.t.}\,\,\,\,  0\leq\alpha_{m,n}\leq 1,\,\, \forall \, \, m,n\in \llbracket 1,K \rrbracket,\label{constraint:P2.1:budget1}\\
&\quad \,\,\,\,\, 0\leq P^{\rm d}_{m,n}\leq \bar{P}^{\rm d},\,\, \forall \, \, m,n\in \llbracket 1,K \rrbracket,\label{constraint:P2.1:budget2}\\
&\quad \,\,\,\,\,R_{m}\left(\alpha_{m,n},P^{\rm d}_{m,n}\right)\geq R_{\rm th} ,\,\,\forall \, \, m,n\in \llbracket 1,K \rrbracket,\label{constraint:P2.1:minimumRateFull1}\\
&\quad \,\,\,\,\,R_{n}\left(\alpha_{m,n},P^{\rm d}_{m,n}\right)\geq R_{\rm th} ,\,\,\forall \, \, m,n\in \llbracket 1,K \rrbracket,\label{constraint:P2.1:minimumRateFull2}
\end{align}
\end{subequations}
for each pair of AUs $(m,n) \in \llbracket 1,K \rrbracket^2$. The inner problem $\mathcal{P}_{\rm inner}$ is a power control problem for a given pair of AUs and it defines a new set of feasible solutions for the outer problem $\mathcal{P}_{\rm outer}$, which is a linear assignment problem. Nevertheless, since we need to solve problem $\mathcal{P}_{\rm inner}$ for all possible combinations of strong and weak AUs, a computational efficient approach for solving $\mathcal{P}_{\rm inner}$ needs to be investigated, which is the focus of the following section.
\subsection{Power Control For a C-NOMA Pair}
\label{sec:Solution-Approach}
In this section, our objective is solving problem $\mathcal{P}_{\rm inner}$ for both HD and FD C-NOMA cases. We consider first the HD C-NOMA case and then we investigate the case of FD C-NOMA.
\subsubsection{HD C-NOMA} In this part, we consider the case of HD C-NOMA and hence the sum-rate $R_{m,n}\left(\alpha_{m,n},P^{\rm d}_{m,n}\right) = R^{\rm H}_{m}\left(\alpha_{m,n},P^{\rm d}_{m,n}\right)+R^{\rm H}_{n}\left(\alpha_{m,n},P^{\rm d}_{m,n}\right)$. Let us investigate first the feasibility conditions of the inner problem $\mathcal{P}_{\rm inner}$. These conditions are detailed in the following theorem.
\begin{theorem}
\label{theorem 1}
For the case of HD C-NOMA, the inner problem $\mathcal{P}_{\rm inner}$ is feasible if and only if the following conditions hold.
\begin{subequations}
\begin{align}
&P_{\rm BS} \geq \frac{(\delta_{\rm th}^{\rm H})^2+\delta_{\rm th}^{\rm H}}{\gamma_{m}},\label{Cond-HD1}
\\
&\bar{P}^{\rm d} \geq P_{m,n}^{\min}\label{Cond-HD2}.
\end{align}
\end{subequations}
where $\delta_{\rm th}^{\rm H}= 2^{2R_{\rm th}}-1$ and 
\begin{equation}
P_{m,n}^{\min} = \frac{((\delta_{\rm th}^{\rm H})^2+\delta^{\rm H}_{\rm th}(\gamma_{m}+P_{\rm BS}\gamma_{n})-P_{\rm BS}\gamma_{n}\gamma_{m}}{\gamma^{\rm d}_{m,n}(\delta^{\rm H}_{\rm th}\gamma_{n}+\gamma_{m})}.
\end{equation}
\end{theorem}
\begin{proof}
Please see Appendix \ref{proof:HD}
\end{proof}
The conditions \eqref{Cond-HD1} and \eqref{Cond-HD2} are necessary to ensure that there exists a power control decision for the inner problem $\mathcal{P}_{\rm inner}$ in case of HD C-NOMA. In addition, these conditions describe the relations between the available power budgets at the BS and at the strong user, the channel gains of both users and the minimum QoS requirements that need to be satisfied in order to make the inner problem $\mathcal{P}_{\rm inner}$ feasible. However, in the case where at least one of the conditions can not be fulfilled, the two possible options to guarantee the required QoS for both users is either reducing the minimum required rate $R_{\rm th}$ or increasing the power budgets $P_{\rm BS}$ and $\bar{P}^{\rm d}$. Now, assuming that the inner problem $\mathcal{P}_{\rm inner}$ is feasible, the optimal power control scheme is presented in the following theorem.
\begin{theorem}
\label{theorem 2}
For the case of HD C-NOMA, the optimal power control scheme, which is defined by the couple $\left(\alpha_{m,n}^*,P^{{\rm d}^*}_{m,n}\right)$, for the inner problem $\mathcal{P}_{\rm inner}$, is expressed as
\begin{equation}
\left(\alpha_{m,n}^*,P^{{\rm d}^*}_{m,n}\right) = \left\{ 
\begin{aligned}
&\left(\frac{(\delta_{\rm th}^{\rm H}-\bar{P}^{\rm d}\gamma^{\rm d}_{m,n})(P_{\rm BS}\gamma_{n}+1)}{P_{\rm BS}\gamma_{n}(\delta_{\rm th}^{\rm H}-\gamma^{\rm d}_{m,n}\bar{P}^{\rm d}+1)} ,\bar{P}^{\rm d}\right) \quad &\text{if}\,\,\bar{P}^{\rm d}\in \left[P_{m,n}^{\min},P_{m,n}^{\text{int}}\right] \\ 
&\left(\frac{\delta_{\rm th}^{\rm H}(P_{\rm BS}\gamma_{m}+1))}{P_{\rm BS}\gamma_{m}({\delta_{\rm th}^{\rm H}+1)}} ,P_{m,n}^{\text{int}}\right) \quad &\text{if}\,\,\bar{P}^{\rm d}\in\left[P_{m,n}^{\text{int}},+\infty\right],
\end{aligned}
\right.
\label{eq:optimal_power-HD}
\end{equation}
where,
\begin{equation}
P_{m,n}^{\text{int}} = \frac{P_{\rm BS}}{\gamma^{\rm d}_{m,n}}\frac{\delta_{\rm th}^{\rm H}(\delta_{\rm th}^{\rm H}+1)(\gamma_{m}-\gamma_{n})}{P_{\rm BS}\gamma_{m}(\delta_{\rm th}^{\rm H}+P_{\rm BS}\gamma_{n}+1)-\delta_{\rm th}^{\rm H}P_{\rm BS}\gamma_{n}}
\end{equation}
\end{theorem}
\begin{proof}
Please see Appendix \ref{proof:HD}
\end{proof}
\subsubsection{FD C-NOMA}
In this part, we consider the case of FD C-NOMA and hence $R_{m,n}\left(\alpha_{m,n},P^{\rm d}_{m,n}\right) = R^{\rm F}_{m}\left(\alpha_{m,n},P^{\rm d}_{m,n}\right)+R^{\rm F}_{n}\left(\alpha_{m,n},P^{\rm d}_{m,n}\right)$. In this case, the feasibility conditions of the inner problem $\mathcal{P}_{\rm inner}$ are detailed in the following theorem.
\begin{theorem}
\label{theorem 3}
For the case of FD C-NOMA, the inner problem $\mathcal{P}_{\rm inner}$ is feasible if and only if the following conditions hold.
\begin{subequations}
\label{Cond-FD}
\begin{align}
    &\text{Condition 1:} \,\, \Delta_1 \geq 0 \land b1<0 \land b_3 \leq b_2 \land \bar{P}^{\rm d}\geq b_3, \\ 
    &\text{Condition 2:} \,\, \Delta_1 \geq 0 \land b1\geq 0, \\
    &\text{Condition 3:} \, \, \Delta_1 < 0, 
\end{align}
\end{subequations}
where the parameters $\Delta_1$, $b_1$, $b_2$, and  $b_3$ are defined in \eqref{Delta_1}-\eqref{b3} on top of next page, in which $\delta^{\rm F}_{\rm th} = 2^{R_{\rm th}}-1$ and $\Delta_2$ is expressed as shown in $\eqref{Delta_2}$ on top of next page.
\end{theorem}
\begin{proof}
Please see Appendix \ref{proof:FD}
\end{proof}
\begin{figure*}[t] 
\begin{equation}
\begin{split}
    \Delta_1 &= \left[P_{\rm BS}\left(\gamma_{m}\gamma^{\rm d}_{m,n}-\gamma_{n}\gamma^{\rm SI}_m(\delta_{\rm th}^{\rm F})^2-\gamma_{n}\gamma^{\rm SI}_m\delta_{\rm th}^{\rm F}+\gamma_{n}\gamma^{\rm d}_{m,n}\delta_{\rm th}^{\rm F}\right)\right]^2\\ 
    &-4P_{\rm BS}^3\gamma_{m}\gamma^{\rm SI}_m\gamma^{\rm d}_{m,n}\delta_{\rm th}^{\rm F}\left[P_{\rm BS}\gamma_{m}\gamma_{n}-\gamma_{m}\delta_{\rm th}^{\rm F}-\gamma_{n}(\delta_{\rm th}^{\rm F})^2-\gamma_{n}\delta_{\rm th}^{\rm F}\right],
    \label{Delta_1}
\end{split}
\end{equation}
\noindent\makebox[\linewidth]{\rule{\textwidth}{0.4pt}}
\begin{equation}
         b_1= \frac{\biggl[-\sqrt{\Delta_1}+P_{\rm BS}^2\gamma_{m}\gamma_{n}\gamma^{\rm d}_{m,n}+P_{\rm BS}(\gamma_{m}\gamma^{\rm d}_{m,n}\delta_{\rm th}^{\rm F}+\gamma_{m}\delta_{\rm th}^{\rm F}+\gamma_{m}\gamma^{\rm d}_{m,n}+\gamma_{n}\gamma^{\rm SI}_m(\delta_{\rm th}^{\rm F})^2+\gamma_{n}\gamma^{\rm SI}_m\delta_{\rm th}^{\rm F}-\gamma_{n}\gamma^{\rm d}_{m,n}\delta_{\rm th}^{\rm F})\biggl]}{(2P_{\rm BS}^2\gamma_{m}\gamma^{\rm SI}_m\gamma^{\rm d}_{m,n}\delta_{\rm th}^{\rm F})} \\ 
    \label{b1}
\end{equation}
\noindent\makebox[\linewidth]{\rule{\textwidth}{0.4pt}}
\begin{equation}
    b_2=\frac{\biggl[\Delta_2+P_{\rm BS}^2\gamma_{m}\gamma_{n}\gamma^{\rm d}_{m,n}+P_{\rm BS}(\gamma_{m}\gamma^{\rm d}_{m,n}\delta_{\rm th}^{\rm F}+\gamma_{m}\delta_{\rm th}^{\rm F}+\gamma_{m}\gamma^{\rm d}_{m,n}+\gamma_{n}\gamma^{\rm SI}_m(\delta_{\rm th}^{\rm F})^2+\gamma_{n}\gamma^{\rm SI}_m\delta_{\rm th}^{\rm F}-\gamma_{n}\gamma^{\rm d}_{m,n}\delta_{\rm th}^{\rm F})\biggl]}{(2P_{\rm BS}^2\gamma_{m}\gamma^{\rm SI}_m\gamma^{\rm d}_{m,n}\delta_{\rm th}^{\rm F})},
    \label{b2}
\end{equation}
\noindent\makebox[\linewidth]{\rule{\textwidth}{0.4pt}}
\begin{equation}
    b_3=\frac{\biggl[\sqrt{\Delta_1}+P_{\rm BS}^2\gamma_{m}\gamma_{n}\gamma^{\rm d}_{m,n}+P_{\rm BS}(\gamma_{m}\gamma^{\rm d}_{m,n}\delta_{\rm th}^{\rm F}+\gamma_{m}\delta_{\rm th}^{\rm F}+\gamma_{m}\gamma^{\rm d}_{m,n}+\gamma_{n}\gamma^{\rm SI}_m(\delta_{\rm th}^{\rm F})^2+\gamma_{n}\gamma^{\rm SI}_m\delta_{\rm th}^{\rm F}-\gamma_{n}\gamma^{\rm d}_{m,n}\delta_{\rm th}^{\rm F})\biggl]}{(2P_{\rm BS}^2\gamma_{m}\gamma^{\rm SI}_m\gamma^{\rm d}_{m,n}\delta_{\rm th}^{\rm F})},
    \label{b3}
\end{equation}
\noindent\makebox[\linewidth]{\rule{\textwidth}{0.4pt}}
\begin{equation}
\begin{split}
    \Delta_2 &= \left[P_{\rm BS}^2\gamma_{m}\gamma_{n}\gamma^{\rm d}_{m,n}+P_{\rm BS}(\gamma_{m}\gamma^{\rm d}_{m,n}\delta_{\rm th}^{\rm F}+\gamma_{m}\delta_{\rm th}^{\rm F}+\gamma_{m}\gamma^{\rm d}_{m,n}+\gamma_{n}\gamma^{\rm SI}_m(\delta_{\rm th}^{\rm F})^2+\gamma_{n}\gamma^{\rm SI}_m\delta_{\rm th}^{\rm F}-\gamma_{n}\gamma^{\rm d}_{m,n}\delta_{\rm th}^{\rm F})\right]^2 \\ 
    &-4P_{\rm BS}^3\gamma_{m}\gamma^{\rm SI}_m\gamma^{\rm d}_{m,n}\delta_{\rm th}^{\rm F}\left[\gamma{\rm s}(\delta_{\rm th}^{\rm F})^2+\gamma{\rm s}\delta_{\rm th}^{\rm F}-\gamma{\rm w}(\delta_{\rm th}^{\rm F})^2-\gamma{\rm w}\delta_{\rm th}^{\rm F}\right]
    \label{Delta_2}
\end{split}
\end{equation}
\noindent\makebox[\linewidth]{\rule{\textwidth}{0.4pt}}
\end{figure*}
Theorem \ref{theorem 3} provides the feasibility conditions of the inner problem $\mathcal{P}_{\rm inner}$, which ensure that there exists at least a value of $\alpha_{m,n}$ and $P^{\rm d}_{m,n}$ that satisfy the constraints \eqref{constraint:P2.1:minimumRateFull1}-\eqref{constraint:P2.1:minimumRateFull2}. Now, based on the results of theorem 3, the closed-form solution of the optimal power control scheme of problem for the case of FD C-NOMA is given in the following theorem.
\begin{theorem}
\label{theorem 4}
For FD C-NOMA, the optimal power control scheme of the inner problem $\mathcal{P}_{\rm inner}$ is given by the couple $\left(\alpha_{m,n}^*,P^{{\rm d}^*}_{m,n}\right)$, where $P^{{\rm d}^*}_{m,n}$ is expressed as 
\begin{equation}
    \begin{aligned}
    &P^{{\rm d}^*}_{m,n} = \\ 
    &\left\{ 
    \begin{aligned}
    &\bar{P}^{\rm d} \times \mathbb{1} \left(0 \leq \bar{P}^{\rm d} \leq b_2 \right) + b_2 \times \mathbb{1} \left(b_2 \leq \bar{P}^{\rm d} \right), \quad &&\text{if} \,\,\Delta_1 < 0, \\
    &\bar{P}^{\rm d} \times \mathbb{1} \left(0 \leq \bar{P}^{\rm d} \leq \max(b_1,b_2) \right) + \max(b_1,b_2) \times \mathbb{1} \left(\max(b_1,b_2) \leq \bar{P}^{\rm d} \right), \quad &&\text{if} \, \,\Delta_1 \geq 0 \land b_1 \geq 0,\\
    &\bar{P}^{\rm d} \times \mathbb{1} \left(\max(0,b_3) \leq \bar{P}^{\rm d} \leq b_2 \right) + b_2 \times \mathbb{1} \left(b_2 \leq \bar{P}^{\rm d} \right), \quad &&\text{if} \, \,\Delta_1 \geq 0 \land b_1 \leq 0 \land b_3\leq b_2,
    \end{aligned}
    \right.
    \end{aligned}
    \label{eqn:power}
\end{equation}
in which $\mathbb{1} \left( \cdot \right)$ denotes the indicator function, and $\alpha_{m,n}^* = f \left(P^{{\rm d}^*}_{m,n} \right)$, such that the function $f$ is given, for all $x \in \mathbb{R}$, by
\begin{equation}
    f(x) = \frac{\delta^{\rm F}_{\rm th}(\gamma^{\rm SI}_mx+\gamma_{m}P_{\rm BS}+1)}{\gamma_{m}P_{\rm BS}(\gamma_{\rm th}+1)}.
\end{equation}
\end{theorem}
\begin{proof}
Please see Appendix \ref{proof:FD}.
\end{proof}
\subsection{Pairing Policy and Proposed Algorithm}
After deriving the optimal power control schemes for each pair and the corresponding achievable rate, we can proceed to apply the Hungarian method to determine the optimal pairing configurations. Let $\textbf{g}$ be the $2K\times 1$ vector that contains the channel gains from the BS to the AUs in a way that are sorted in an ascending order, where entries $g_1,g_2,...,g_K$ represent the channel gains of the weak AUs and $g_{K+1},g_{K+2},...,g_{2K}$ represent the channel gains of the strong AUs. In addition, let $\mathbf{D}$ be the $K \times K$ matrix that contains the D2D channel gains and let $\textbf{s}$ be the $K \times 1$ vector that contains the SI channel gains. The input of Hungarian algorithm is a $K\times K$ cost matrix $\textbf{C}$ and its output is the pairing matrix $\textbf{B}^*$ where $\textbf{B}^{*}(m,n)= 1$ indicates that the AU $m$ in the strong AUs set $\mathcal{S}$ is paired with the AU $n$ in the weak AUs set $\mathcal{W}$, and $\textbf{B}^{*}(m,n)=0$ otherwise. In our algorithm, we define the cost of pairing two users as the opposite value of the sum rate obtained by solving the inner problem $\mathcal{P}_{\rm inner}$. After the computation of $\textbf{C}$, the Hungarian algorithm is applied to solve the outer problem $\mathcal{P}_{\rm outer}$. The proposed algorithm is highlighted in Algorithm~\ref{alg:algorithm1}. After computing the matrix $\textbf{B}^*$, it is straightforward to solve problem $\mathcal{P}$ since the optimal value $\alpha_{m,n}^*$ and $P_{m,n}^{{\rm d}^*}$ within each pairs have been found in step 2 of Algorithm~\ref{alg:algorithm1}. After applying algorithm \ref{alg:algorithm1}, we obtain the optimal pairing policy and the optimal power control scheme within each pair, which is the final solution for problem $\mathcal{P}$. 
\begin{algorithm}[t]
\caption{Low-complexity algorithm.}
\begin{algorithmic} 
\STATE 1. \textbf{Estimate} channel gain vector $\textbf{g}$, D2D channel gain matrix $\textbf{D}$, and the SI channel gain vector $\textbf{s}$.
\STATE 2. \textbf{Compute} the optimal rate for each pairing configuration $R_{m,n}^*$ following theorem \ref{theorem 2} for HD case and theorem \ref{theorem 4} for FD case.
\STATE 3. \textbf{Compute} the cost matrix $\textbf{C}$ with $\mathbf{C}(m,n) = -R_{m,n}^*$.
\STATE 4. \textbf{Solve} the optimal pairing matrix $\textbf{B}^*$ using the Hungarian algorithm  with cost matrix $\textbf{C}$.
\end{algorithmic}
\label{alg:algorithm1}
\end{algorithm} 
\subsection{Mode Selection}
\label{subsec:mode}
Mode selection has been proposed in previous works on relaying systems in general and C-NOMA based system in particular \cite{Liu:C-NOMA:2users,Li:ModeSelection:Letter}. Mode selection assumes that the user device have the capability of both FD and HD relaying and can switch from one mode to the other (hybrid devices). The idea of mode selection stems from the fact that FD relaying does not necessarily perform better than HD relaying especially with the presence of high SI.
Since we can compute the resulting optimal sum-rate value for both FD C-NOMA and HD C-NOMA based on theorem \ref{theorem 2} and theorem \ref{theorem 4}, determining the optimal mode is straightforward. In fact, we just need to select the mode that gives the highest resulting sum-rate from the power control schemes derived in these theorems.
\subsection{Complexity Analysis}
Complexity is worth bringing into discussion, since the complexity of the proposed method seems extreme and it requires the computation of the rate values for all possible pairing configurations. However, it is worth mentioning that, due to the closed-form solution obtained in Theorems \ref{theorem 2} and \ref{theorem 4}, the computational complexity of obtaining the optimal sum rate for a given pair is approximately $\mathcal{O}(1)$. Thus, even for all the possible configurations, the computational complexity of obtaining all the sum rates is $\mathcal{O}(K^2)$. With the addition of the Hungarian method, the total computational complexity of our proposed algorithm is approximately $\mathcal{O}(K^2+K^3)$ or $\mathcal{O}(K^3)$ for large values of $K$. Clearly, the overall computational complexity depends more on the Hungarian method than the computation of the cost matrix $\textbf{C}$.
\section{Simulation Results}
\label{sec:Simulation Results}
In this section, we will validate the proposed scheme. We assume that the channel gains of the strong AU $h_{m}$, the weak AU $h_{n}$, the D2D link $h^{\rm d}_{m,n}$, and the SI $h^{\rm SI}_m$ follows independent Rayleigh distributions with scale parameters, $\sigma_{\rm s}$, $\sigma_{\rm w}$, $\sigma_{\rm d}$ and $\sigma_{\rm SI}$, respectively. Therefore, their associated squares $\gamma_{m}$, $\gamma_{n}$, $\gamma^{\rm d}_{m,n}$ and $\gamma^{\rm SI}_m$ follow independent exponential distributions with means $\lambda_{\rm s} = \left(2 \sigma_{\rm s} \right)^{-1}$, $\lambda_{\rm w} = \left(2 \sigma_{\rm w} \right)^{-1}$, $\lambda_{\rm d} = \left(2 \sigma_{\rm d} \right)^{-1}$, and $\lambda_{\rm SI} = \left(2 \sigma_{\rm SI} \right)^{-1}$, respectively. The required QoS is defined by the minimum achievable rate threshold that is given by $R_{\rm th} = 1$ [bps/Hz]. Simulation results are performed over $10^5$ independent Monte-Carlo trials on the channel gain realizations. \\
\begin{figure}[t]
    \centering
    \includegraphics[width= 0.55\textwidth]{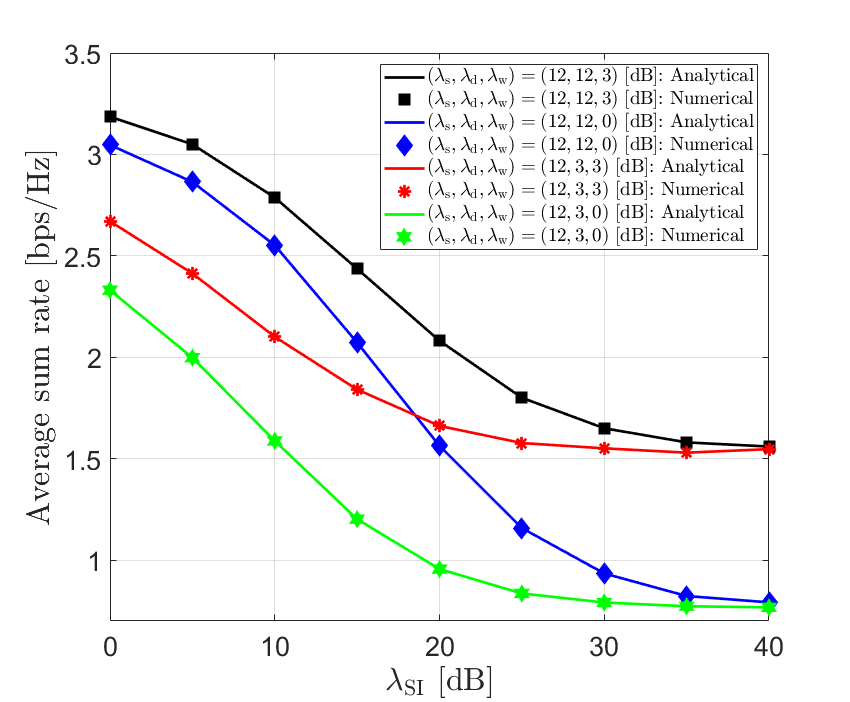}
    \caption{Analytical and numerical average achievable sum-rate versus the average self interference channel $\lambda_{\rm SI}$.}     
    \label{fig:aVSn}
\end{figure}
\indent Fig.~\ref{fig:aVSn} presents the analytical and numerical average sum rate for one pair of users with a FD C-NOMA transmission scheme versus the average self interference channel $\lambda_{\rm SI}$ for different means of the channel gain of the strongest AU $\lambda_{\rm s}$, the weakest AU $\lambda_{\rm w}$ and the D2D link $\lambda_{\rm d}$. The analytical results are obtained through the closed-form power control scheme derived in theorem 4 whereas the numerical results are obtained by solving problem $\mathcal{P}_{\rm inner}$ using an off-the-shelf optimization solver.\footnote{The adopted solver is fmincon, which is a predefined matlab solver \cite{ebbesen2012generic}. In addition, 100 distinct initial points were generated in order to converge to the optimal solution.} This figure shows that the analytical results match perfectly the numerical results, which validate the optimality of the power control scheme derived in theorem 4. In addition, Fig.~\ref{fig:aVSn} shows that when the mean of the SI channel increases, the average sum rate decreases. This observation is expected since as shown in equations (10) and (11), the achievable rates of the strongest AU and the weakest AU are decreasing with respect to $\lambda_{\rm SI}$. On the other hand, by comparing the cases where $\left(\lambda_{\rm s}, \lambda_{\rm d}, \lambda_{\rm w} \right) = \left(12,12,3 \right)$ [dB] and $\left(12,3,3 \right)$ [dB], this figure shows that, when the mean of the SI channel increases, the potential of the D2D relaying decreases. The same observation holds when comparing the cases $\left(\lambda_{\rm s}, \lambda_{\rm d}, \lambda_{\rm w} \right) = \left(12,12,0 \right)$ [dB] and $\left(12,3,0 \right)$ [dB]. This observation makes sense as well, since 
when $\lambda_{\rm SI}$ increases, the SI at the strongest AU increases, and thus, the only way to alleviate its effect is decreasing the transmit D2D relaying power $P^{\rm d}_{m,n}$. This makes C-NOMA converges to the conventional NOMA. \\
\begin{figure}[t]
    \centering
 \includegraphics[width= 0.55\textwidth]{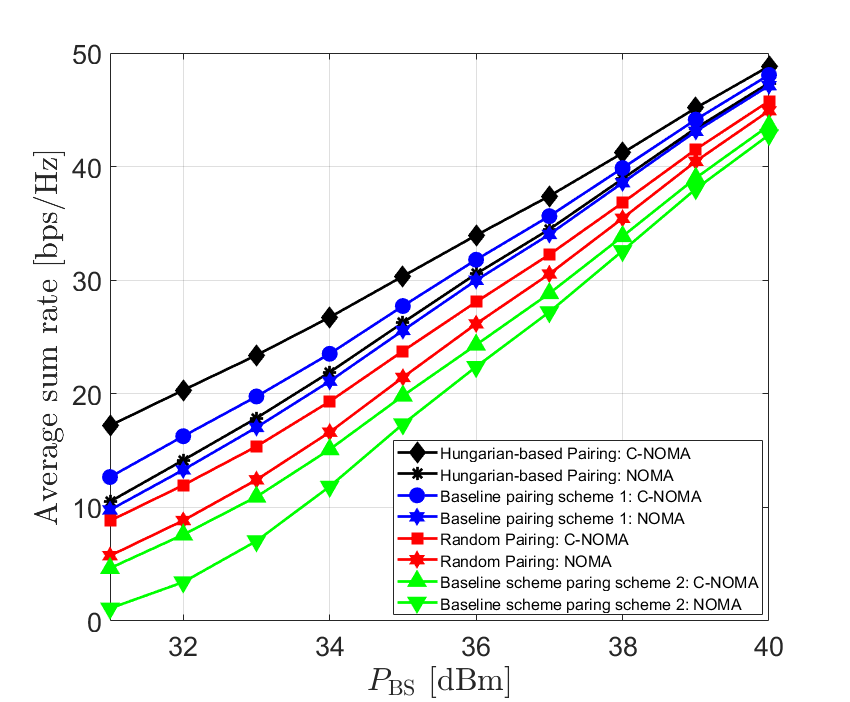}
  \caption{Average achievable sum-rate versus the power budget at the base station $P_{\rm BS}$ for C-NOMA and NOMA with different pairing schemes ($\lambda_{\rm s} = 10$ dB, $\lambda_{\rm w} =  0$ dB, $\lambda_{\rm d} =  6$ dB, $\lambda_{\rm SI} =  0$ dB, $\bar{P}^{\rm d}$ = 30 dBm)}  
  \label{fig:compare1}
\end{figure}
\indent Fig.~\ref{fig:compare1} presents the average sum rate achieved by the FD C-NOMA and conventional NOMA with different pairing schemes versus the power budget at the base station $P_{\rm BS}$. For conventional NOMA, different pairs of AUs are served via OMA, where the AUs within each pair are served via NOMA with no cooperation between users. Besides, we compare the proposed user pairing policy with three different pairing schemes. For baseline pairing scheme 1, we pair the $k$th AU $(1\leq k \leq K)$ with the $(2K-k+1)$th AU, e.g, the weakest AU is paired with the strongest AU, the second weakest is paired with the second strongest, and so on \cite{Letter:NOMA:Pairing}. For baseline pairing scheme 2, we pair the $k$th AU $(1\leq k \leq K)$ with the $(K+k)$th AU, e.g, the $k$th weakest AU in the set of weak AUs $\mathcal{W}$ is paired with the $k$th weakest AU in the set of strong AUs $\mathcal{S}$. For random pairing scheme, two randomly selected AU in each set are paired with each other. Fig.~\ref{fig:compare1} shows that the proposed user pairing policy outperforms the baseline schemes. In addition, this figure shows that enabling cooperation between users within each pair improves the performance of the system. \\
\begin{figure}[t]
\centering
  \includegraphics[width=0.55\textwidth, height = 0.5\textwidth]{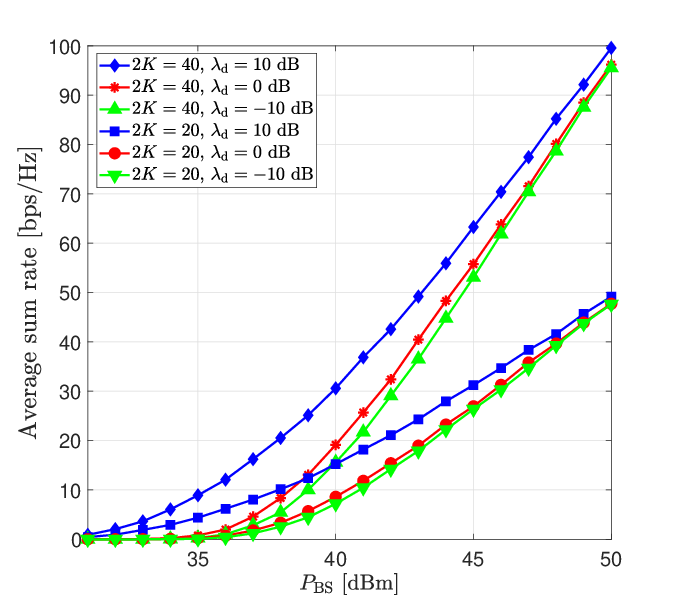}
  \caption{Average sum rate versus the power budget at the base station $P_{\rm BS}$ for C-NOMA and different means of the D2D channel gain $\lambda_{\rm d}$ ($\lambda_{\rm s} = 10$ dB, $\lambda_{\rm w} =  6$ dB $\lambda_{\rm SI} =  6$ dB, $\bar{P}^{\rm d}$ = 30 dBm)}
  \label{fig:number_users}
\end{figure}
\indent Fig.~\ref{fig:number_users} presents the average sum rate achieved by FD C-NOMA and the Hungarian method versus the power budget at the base station $P_{\rm BS}$ when the total number of users is $2K = 20$ and $2K = 40$ and for different values of the mean D2D channel $\lambda_D$. This figure shows that the average sum rate increases when the total number of users increases and when the mean D2D channel $\lambda_D$ increases. This observation is expected since, as shown in equation (13) the achievable rate of the weak AU increases when $\lambda_D$ increases. Note that, in Fig.~\ref{fig:aVSn}, Fig.~\ref{fig:compare1} and Fig.~\ref{fig:number_users}, the performance of HD C-NOMA follows similar pattern to the one of FD C-NOMA and this is why it is omitted here. \\
\begin{figure}
    \centering
    \includegraphics[width=0.55\textwidth, height = 0.5\textwidth]{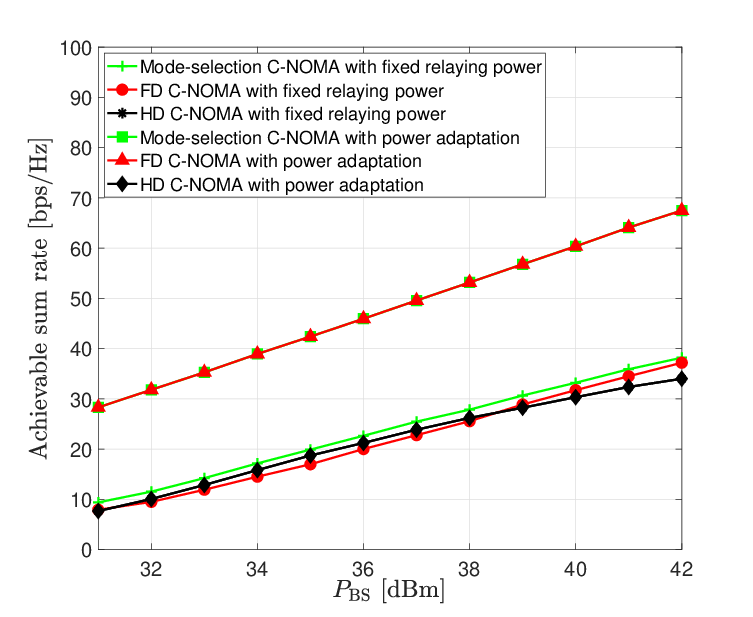}
    \caption{Comparisons between different relaying modes of cooperation in C-NOMA system for fixed and adaptive relaying power control with different power budget at the BS $P_{\rm BS}$ ($\lambda_{\rm s} = 10$ dB, $\lambda_{\rm w} =  6$ dB, $\lambda_{\rm d} =  \lambda_{\rm SI} =  6$ dB).}
    \label{fig:compare3modes}
\end{figure}
\indent In Fig. \ref{fig:compare3modes}, we evaluate the performance of the system when the AUs can operate under different relaying modes. In total, we compare between three relaying modes, which are HD relaying, FD relaying and hybrid relaying with mode selection as presented in subsection \ref{subsec:mode}. In addition, we consider the cases where, within each pair, either the strong AU forwards the message of the weak AU with all its available D2D power budget $\bar{P}^{\rm d}$, i.e., without power adaptation, or it employs the closed-form solution of the optimal transmit D2D power $P^{{\rm d}^*}_{m,n}$ obtained in theorem 2 for the HD case and in theorem 4 for the FD case. Three observations can be remarked from Fig. \ref{fig:compare3modes}. First, with power adaptation and assuming the given parameters, the best relaying mode is FD mode, which explains why C-NOMA with mode selection and FD C-NOMA have the same performance. Second, the HD C-NOMA is not affected by the power adaptation at each strong AU, since even with power adaptation, HD C-NOMA will always choose to transmit with maximum power. This observation is expected because HD C-NOMA operates without the induction of SI at the strong users. Thus, the higher power the strong AUs can transmit with, the higher the achievable sum rate is. The third observation is that, for the case when there is no power adaptation at the strong AUs, HD C-NOMA can outperform FD C-NOMA although FD C-NOMA can operate over the whole subchannel assigned. Obviously, without a proper power adaption that is based on the channel conditions, the transmit power at the strong AUs will adversely affect its achievable rate due to SI. This trend is highly observable when the BS power budget $P_{\rm BS}$ is lower than D2D power budget $\bar{P}^{\rm d}$, which means that the impact of the D2D relaying transmission becomes comparable to that of the downlink access transmission. In this situation, the relaying mode becomes an important feature since the SI channel gain $h_{\rm SI}$ and the adopted D2D transmit power $P^{\rm d}_{m,n}$ will significantly impact the overall achievable sum rate. When the power $P_{\rm BS}$ becomes higher, the downlink access transmission can be enough to fulfil the minimum QoS constraints and the D2D cooperative transmission becomes less influential on the sum rate of the overall system; therefore FD C-NOMA will prevail HD C-NOMA since the downlink transmission can occupy double the bandwidth resource. \\
\begin{figure}[t]
    \centering
    \includegraphics[width=0.55\textwidth, height = 0.5\textwidth]{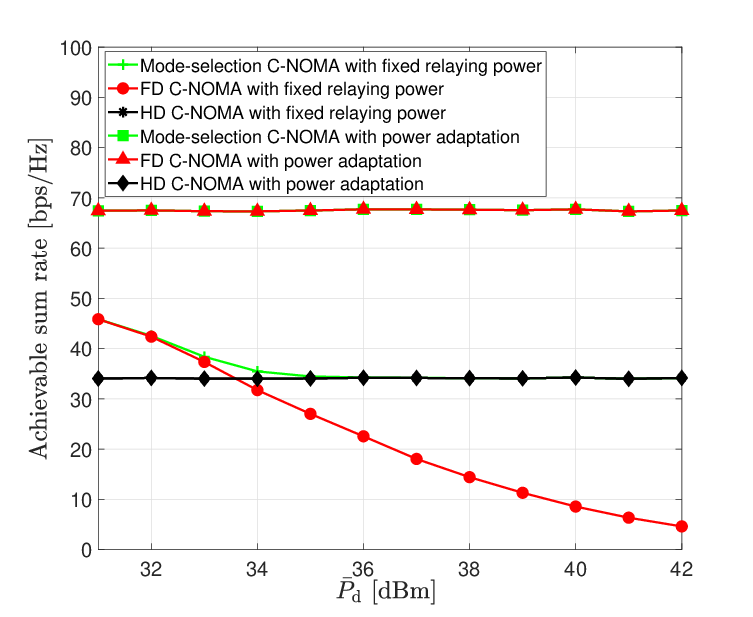}
    \caption{Comparisons between different relaying modes of cooperation in C-NOMA system for fixed and adaptive relaying power control with different power budget at the user device $\bar{P}^{\rm d}$ ($\lambda_{\rm s} = 10$ dB, $\lambda_{\rm w} =  6$ dB, $\lambda_{\rm d} =  6$ dB, $\lambda_{\rm SI} =  0$ dB, $P_{\rm BS}= 42$ dBm).}
    \label{fig:compare3modes2}
\end{figure}
\indent Fig. \ref{fig:compare3modes2} presents the average achievable sum rate of the overall C-NOMA system versus the D2D power budget $\bar{P}^{\rm d}$ for the three aforementioned relaying modes, with and without power adaptation at the strong AUs. Similar to Fig. \ref{fig:compare3modes}, for the case with power adaptation, FD C-NOMA and mode-selection C-NOMA have the same performance and stay constant even when the D2D power budget increases. The reason is that the relaying power in these schemes will adapt to the various channel conditions and operate at the optimal power. In other words, it is not necessary to exhaust the D2D transmit power as in the fixed relaying scheme for two reasons. First, it is energy-inefficient for the users' devices. Second, the overall sum rate decreases when there is no power adaptation. Finally, it is also worth noting that for the case of fixed relaying power scheme and low D2D power budget, mode-selection C-NOMA converges to FD C-NOMA, whereas it converges to HD C-NOMA at high D2D power budget. \\
\begin{figure}[t]
    \centering
    \includegraphics[width=0.55\textwidth, height = 0.5\textwidth]{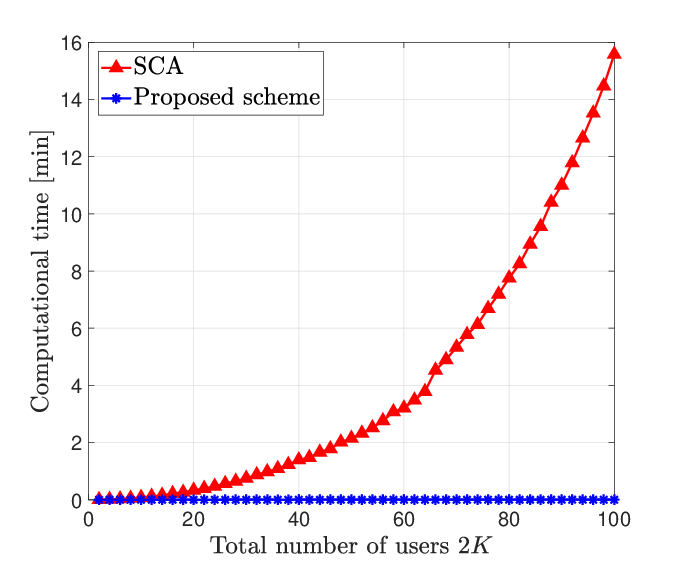}
    \caption{Computational time of the proposed scheme versus that of SCA-based approach .}
    \label{fig:computational time}
\end{figure}
\indent In Fig. \ref{fig:computational time}, we compare the computational time of the proposed algorithm in \ref{alg:algorithm1} with one of the SCA-based scheme. In the SCA-based scheme, the inner problem $\mathcal{P}_{\rm inner}$ is solved using successive linear relaxations of all the non-convex functions until convergence, where each relaxation is the input to an off-the-shelf optimization solver. This method has been repeatedly used for many non-convex optimization problems. Despite giving decent performance, it can be seen that the computational time for SCA-based method scale exponentially when the number of AUs increases while the computational time of our proposed scheme is negligible. For instance, when the number of users is 100, it takes about 16 minutes to compute the optimal result using SCA-based scheme, whereas it takes about 10 ms to compute the optimal solution using the proposed algorithm.
\section{Conclusion}\label{Sec:Conclusion}
This paper studied the performance of C-NOMA in a cellular downlink system, in which users that have the capability of both HD and FD communications can assist the transmissions between the BS and users with poor channel quality through D2D communications. A novel optimization problem that jointly determines the optimal D2D user pairing and the optimal power control scheme was formulated and solved, where the objective is maximizing the achievable sum rate of the whole system while guaranteeing a certain quality of service (QoS) for all users. A two-step policy was proposed to solve the problem in a polynomial time. First, the closed-form expression of the optimal power control scheme that maximizes the sum rate of a given pair of users with a required QoS was derived for both cases of FD and HD communications. Then, using the derived expressions, the Hungarian method was adopted as a user pairing policy, which provided the optimal pairing strategy. The simulation results showed that the proposed scheme always outperforms some existing schemes in the literature. \\ 
\indent This paper presents the first work that investigates the joint user pairing and power control optimization for multi-user C-NOMA. However, it assumes full CSI at the BS. Therefore, the impact of imperfect users CSI at the BS will be studied in future works. 
\section*{Acknowledgment}
Phuc Dinh, Mohamed Amine Arfaoui and Chadi Assi acknowledge the financial supports from Concordia University, Montreal, Quebec, Canada, and from the Fonds Québécois de la Recherche sur la Nature et les Technologies (FQRNT), Quebec, Canada.
\appendices
\section{Power control for HD C-NOMA}
\label{proof:HD}
In this section, we present the proof of theorems \ref{theorem 1} and \ref{theorem 2}. At first, note that for any value $\alpha \in [0,1]$, constraints \eqref{constraint:P2.1:minimumRateFull1} and \eqref{constraint:P2.1:minimumRateFull2} can be equivalently transformed into the following inequalities.
\begin{equation}
    \label{feasible_alpha-HD}
        \max \left(A^{\rm H}_{m,n},B^{\rm H}_{m,n}\left( P^{\rm d}_{m,n} \right)\right) \leq \alpha \leq C^{\rm H}_{m,n},
\end{equation}
where 
\begin{equation}
\begin{split}   
    &A^{\rm H}_{m,n} = \frac{\delta_{\rm th}^{\rm H}(P_{\rm BS}\gamma_{m}+1)}{P_{\rm BS}\gamma_{m}(\delta_{\rm th}^{\rm H}+1)},\\ 
    &B^{\rm H}_{m,n}\left( P^{\rm d}_{m,n} \right) = \frac{(P_{\rm BS}\gamma_{n}+1)(\delta_{\rm th}^{\rm H}-P^{\rm d}_{m,n}\gamma^{\rm d}_{m,n})}{P_{\rm BS}\gamma_{n}(\delta_{\rm th}^{\rm H}+1-P^{\rm d}_{m,n}\gamma^{\rm d}_{m,n})}, \\
    &C^{\rm H}_{m,n} = 1-\frac{\delta_{\rm th}^{\rm H}(\gamma^{\rm SI}_m+1)}{P_{\rm BS}\gamma_{m}}.
\end{split}
\end{equation}
Obviously, $A^{\rm H}_{m,n}$, $B^{\rm H}_{m,n}$, and $C^{\rm H}_{m,n}$ define together the boundaries of the region of the feasible solutions of $\alpha$ as presented in Fig. \ref{fig:HDBounds}. In addition, it can also be observed from equation \eqref{eq:P2.1:obj} that the sum rate is an decreasing function with respect to $\alpha_{m,n}$ and an increasing function with respect to $P^{\rm d}_{m,n}$, which can be confirmed through its partial derivatives. Consequently, to maximize this sum-rate function, the smallest and feasible value of $\alpha_{m,n}$ should be selected. By observing Fig. \ref{fig:HDBounds}, it can be realized that the pre-selected values of $\alpha_{m,n}$ are along the red and the blue lines, depending on the value of the relaying power budget $\bar{P}^{\rm d}$. Let $P_{m,n}^{\min}$ denotes the power value of the intersection point between the black and the red lines, which can be easily found by solving the equation $B^{\rm H}_{m,n}=C^{\rm H}_{m,n}$. Now, in order for the inner problem $\mathcal{P_{\rm inner}}$ to be feasible, it is necessary that the feasibility region (green region in Fig. \ref{fig:HDBounds}) is non-empty. In other words, the value $\bar{P^{\rm d}_{m,n}}$ needs to be greater then $P_{m,n}^{\min}$ and the upper bound $C^{\rm H}$ needs to be greater than $0$. Solving these two inequalities implies the feasibility conditions in \ref{Cond-HD1}-\ref{Cond-HD2}, which completes the proof of theorem \ref{theorem 1}. \\
\indent From the above observation, it is obvious that the value $P^{\rm d}_{m,n}$ should be as large as possible so that value $\alpha$ takes the minimum value. However, after the intersection between the red and the blue lines, which holds at the power value $P_{m,n}^{\text{int}}$, i.e., when $B^{\rm H}_{m,n}=A^{\rm H}_{m,n}$, the value of $\alpha$ will be constant even when $P_{m,n}^{\text{int}} \leq P^{\rm d}_{m,n}$. Thus, when $\bar{P}^{\rm d} \leq P_{m,n}^{\text{int}}$, the optimal value of $P^{\rm d}_{m,n} = \bar{P}^{\rm d}$ and the optimal value of $\alpha = B^{\rm H}_{m,n}\left( \bar{P}^{\rm d} \right)$. However, from a power efficiency perspective, when $P_{m,n}^{\text{int}} \leq \bar{P}^{\rm d}$, the optimal value of $P^{\rm d}_{m,n} = P_{m,n}^{\text{int}}$ and the optimal value of $\alpha = A^{\rm H}_{m,n} = B^{\rm H}_{m,n}\left( P_{m,n}^{\text{int}} \right)$. This scheme implies the closed-form solutions in \eqref{eq:optimal_power-HD}, which completes the proof for theorem \ref{theorem 2}.
\begin{figure}
    \centering
    \includegraphics[width=0.45\textwidth, height = 0.35\textwidth]{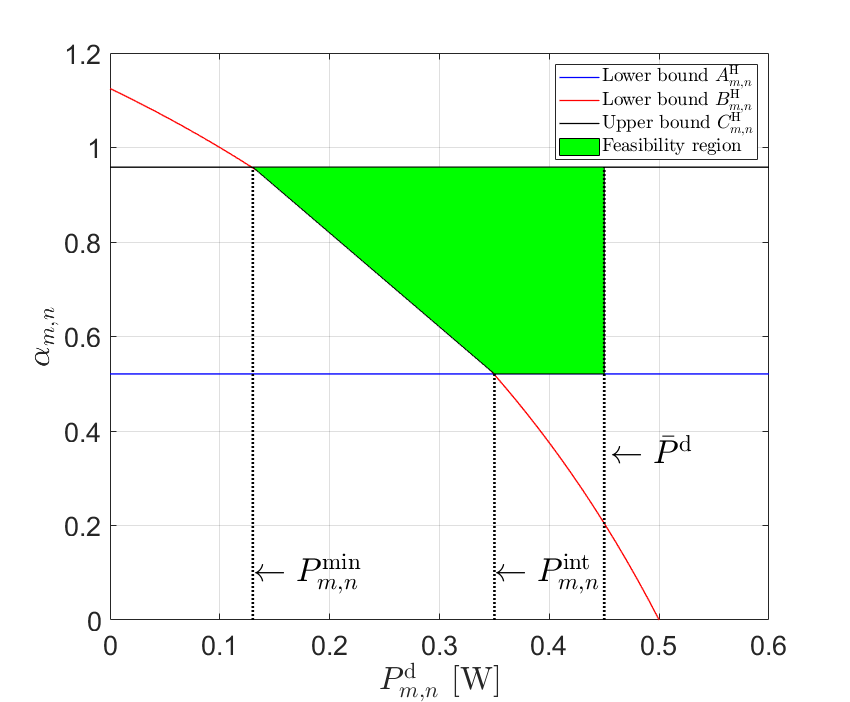}
    \caption{Feasibility region of power control problem for HD C-NOMA.}
    \label{fig:HDBounds}
\end{figure}
\section{Power control for FD C-NOMA}
\label{proof:FD}
\begin{figure*}[t]
\centering     
\subfigure[No intersection.]{\label{fig:app1}\includegraphics[width=0.33\linewidth, height = 0.25\textwidth]{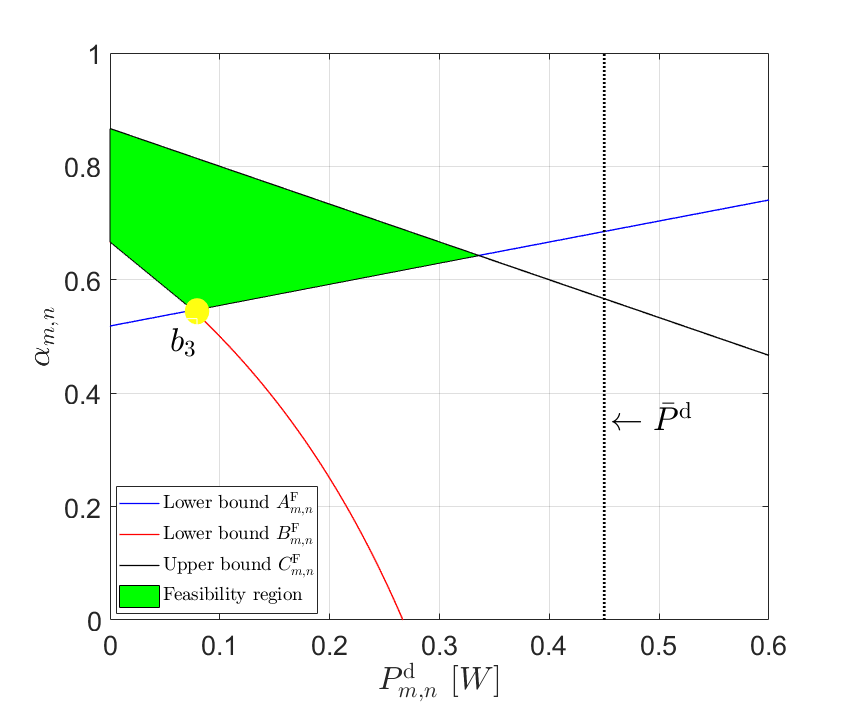}}
\subfigure[Positive lower intersection.]{\label{fig:app2}\includegraphics[width=0.33\linewidth, height = 0.25\textwidth]{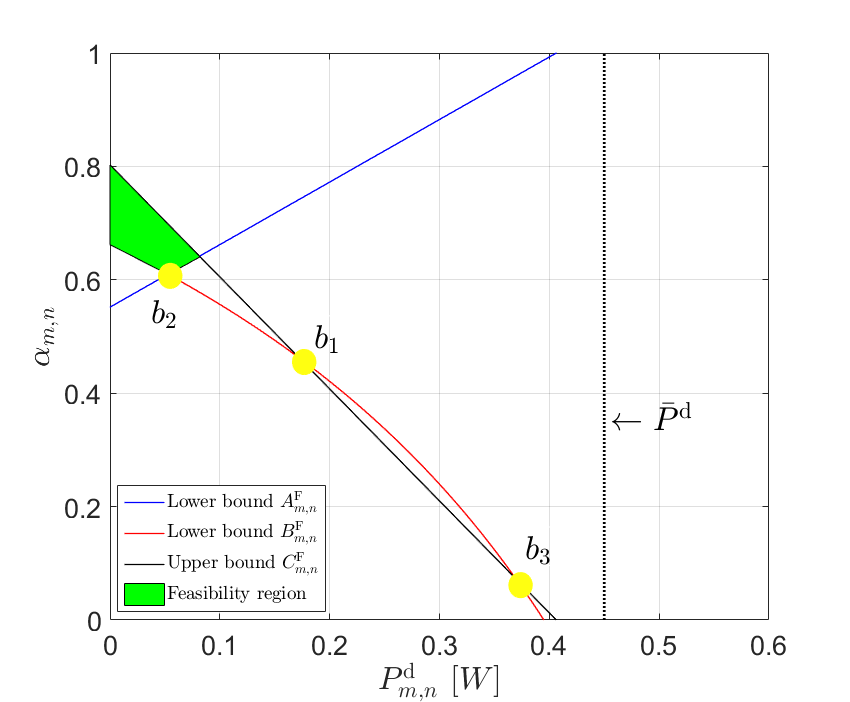}}
\subfigure[Negative lower intersection.]{\label{fig:app3}\includegraphics[width=0.33\linewidth, height = 0.25\textwidth]{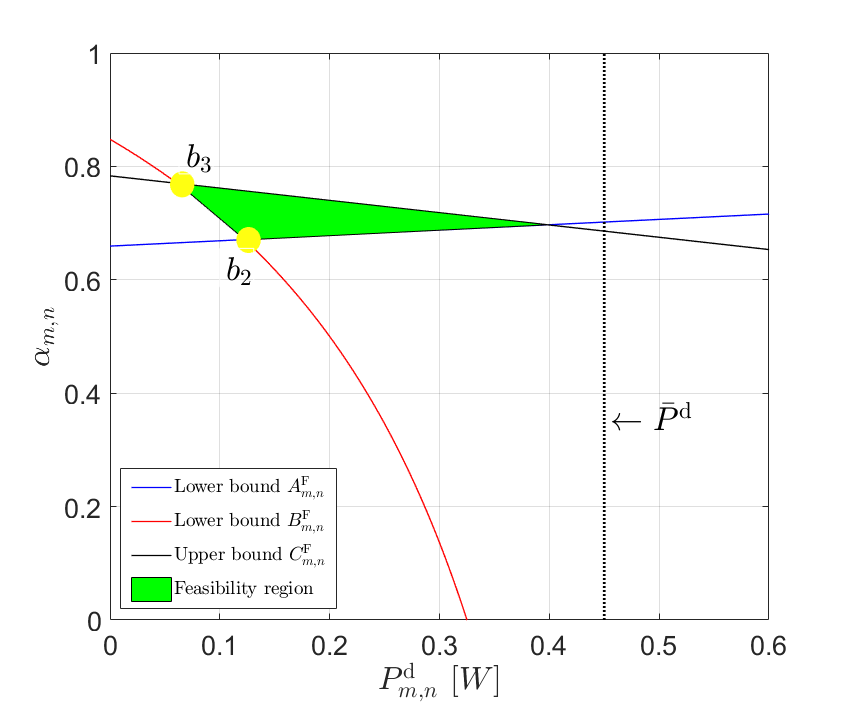}}
\caption{Feasibility region power control problem for FD C-NOMA.}
\label{fig:app}
\end{figure*}
In this section, we present the proof of theorems \ref{theorem 3} and \ref{theorem 4}.At first, note that for any value $\alpha \in [0,1]$, constraints \eqref{constraint:P2.1:minimumRateFull1} and \eqref{constraint:P2.1:minimumRateFull2} can be equivalently transformed into the following inequalities.
\begin{equation}
    \label{feasible_alpha-FD}
        \max \left(A^{\rm F}_{m,n},B^{\rm F}_{m,n}\right) \leq \alpha \leq C^{\rm F}_{m,n},
\end{equation}
where 
\begin{equation}
\begin{split}
    A^{\rm F}_{m,n} \left( P^{\rm d}_{m,n} \right) &= \frac{\delta_{\rm th}^{\rm F}(P_{\rm BS}\gamma_{m}+\gamma^{\rm SI}_m P^{\rm d}_{m,n}+1)}{P_{\rm BS}\gamma_{m}(\delta_{\rm th}^{\rm F}+1)},\\ 
    B^{\rm F}_{m,n} \left( P^{\rm d}_{m,n} \right) &= \frac{(P_{\rm BS}\gamma_{n}+1)(\delta_{\rm th}^{\rm F}-P^{\rm d}_{m,n}\gamma^{\rm d}_{m,n})}{P_{\rm BS}\gamma_{n}(\delta_{\rm th}^{\rm F}+1-P^{\rm d}_{m,n}\gamma^{\rm d}_{m,n})}, \\
    &C^{\rm F}_{m,n} = 1-\frac{\delta_{\rm th}^{\rm F}(\gamma^{\rm SI}_m+1)}{P_{\rm BS}\gamma_{m}}.
\end{split}
\end{equation}
Obviously, $A^{\rm F}$ and $B^{\rm F}$ are functions of the relaying power $P^{\rm d}_{m,n}$, and along with $C^{\rm F}_{m,n}$, they define together the boundaries of the region of the feasible solutions of $\alpha$. Precisely, the feasible region of $\alpha$ is characterized by the relation presented in \eqref{feasible_alpha-FD}. Obviously, there exit feasible solutions if and only if this region is non-empty. In order to characterize this feasible region, we use the geometric representations of \eqref{feasible_alpha-FD} as shown in Fig. \ref{fig:app1},  Fig. \ref{fig:app2} and  Fig. \ref{fig:app1}\ref{fig:app3}, where $b_1$ and $b_3$ are the lower and upper intersections between $B^{\rm F}_{m,n}$ and $C^{\rm F}_{m,n}$, respectively, and $b_2$ is the intersection between $A^{\rm F}_{m,n}$ and $B^{\rm F}_{m,n}$. It can be easily proven that $A^{\rm F}_{m,n}$ intersects with $B^{\rm F}$ at a positive value of $P^{\rm d}_{m,n}$ by solving the equation $A^{\rm F}_{m,n}\left(P^{\rm d}_{m,n}\right) = B^{\rm F}_{m,n}\left(P^{\rm d}_{m,n}\right)$. On the other hand, $B^{\rm F}$ and $C^{\rm F}$ may or may not intersect depending on the value $\Delta_1$ in \ref{Delta_1} that is arisen from solving the equation $B^{\rm F}_{m,n}=C^{\rm F}_{m,n}$. There are three cases where feasible solutions exist as follows. 
\begin{enumerate}
    \item When there is no intersection between $B^{\rm F}_{m,n}$ and $C^{\rm F}_{m,n}$ (Fig. \ref{fig:app1}). 
    \item when the two intersections of $B^{\rm F}_{m,n}$ and $C^{\rm F}_{m,n}$ occur at positive values of $P^{\rm d}_{m,n}$ (Fig. \ref{fig:app2}). 
    \item when the lower intersection of $B^{\rm F}_{m,n}$ and $C^{\rm F}_{m,n}$ occurs at a negative value of $P^{\rm d}_{m,n}$ and their higher intersection (either at a positive or negative value of $P^{\rm d}_{m,n}$) is lower than the intersection between $A^{\rm F}$ and $C^{\rm F}$ (Fig. \ref{fig:app3}).
\end{enumerate}
Based on this, we obtain the feasibility conditions as presented in \eqref{Cond-FD} and it can be easily verified that there is no feasible solutions for $\alpha$ otherwise, which completes the proof for theorem \ref{theorem 3}. \\
\indent After determining the feasibility conditions of $\alpha$, we can proceed to find the optimal power control scheme. First, we note that the inner problem $\mathcal{P}_{\rm inner}$ is non-convex due to constraints \eqref{constraint:P2.1:minimumRateFull1}  and \eqref{constraint:P2.1:minimumRateFull2}. However, it is widely known that the optimal power allocation scheme that maximizes the sum rate follows the water-filling power allocation policy. Thus, we choose the minimum feasible value of $\alpha$ for maximizing the objective function of $\mathcal{P}_{\rm inner}$, and therefore, we the optimal value of $P^{\rm d}_{m,n}$ is the one that achieves the lowest feasible $\alpha$. Based on this and on Fig. \ref{fig:app1},  Fig. \ref{fig:app2} and  Fig. \ref{fig:app1}\ref{fig:app3}, when $\bar{P}^{\rm d} \leq b_2$, the optimal value of $P^{\rm d}_{m,n} = \bar{P}^{\rm d}$ and the optimal value of $\alpha = B^{\rm F}_{m,n} \left( \bar{P}^{\rm d} \right)$. However, when $b_2 \leq \bar{P}^{\rm d}$, the optimal value of $P^{\rm d}_{m,n} = b_2$ and the optimal value of $\alpha = A^{\rm F}_{m,n} \left( b_2 \right) = B^{\rm F}_{m,n} \left( b_2 \right)$, which completes the proof for theorem \ref{theorem 4}.
\bibliographystyle{IEEEtran}
\bibliography{main.bib}

\begin{thebibliography}{10}
\providecommand{\url}[1]{#1}
\csname url@samestyle\endcsname
\providecommand{\newblock}{\relax}
\providecommand{\bibinfo}[2]{#2}
\providecommand{\BIBentrySTDinterwordspacing}{\spaceskip=0pt\relax}
\providecommand{\BIBentryALTinterwordstretchfactor}{4}
\providecommand{\BIBentryALTinterwordspacing}{\spaceskip=\fontdimen2\font plus
\BIBentryALTinterwordstretchfactor\fontdimen3\font minus
  \fontdimen4\font\relax}
\providecommand{\BIBforeignlanguage}[2]{{%
\expandafter\ifx\csname l@#1\endcsname\relax
\typeout{** WARNING: IEEEtran.bst: No hyphenation pattern has been}%
\typeout{** loaded for the language `#1'. Using the pattern for}%
\typeout{** the default language instead.}%
\else
\language=\csname l@#1\endcsname
\fi
#2}}
\providecommand{\BIBdecl}{\relax}
\BIBdecl

\bibitem{Phuc2019joint}
P.~Dinh \emph{et~al.}, ``{Joint user pairing and power control for C-NOMA with
  full-duplex device-to-device relaying},'' in \emph{Proc. IEEE Globecom},
  Waikoloa, HI, USA, Dec. 2019.

\bibitem{Intro1}
C.~V. Networking, ``Cisco global cloud index: Forecast and methodology,
  2016--2021,'' \emph{$\textnormal{White Paper}$}, Cisco, San Jose, CA, USA
  Nov. 2019.

\bibitem{Intro2}
J.~G. Andrews, S.~Buzzi, W.~Choi, S.~V. Hanly, A.~Lozano, A.~C. Soong, and
  J.~C. Zhang, ``{What will 5G be?}'' \emph{IEEE J. on selected areas in
  commun.}, vol.~32, no.~6, pp. 1065--1082, Jun. 2014.

\bibitem{Intro3}
M.~Shafi, A.~F. Molisch, P.~J. Smith, T.~Haustein, P.~Zhu, P.~Silva,
  F.~Tufvesson, A.~Benjebbour, and G.~Wunder, ``{5G: A tutorial overview of
  standards, trials, challenges, deployment, and practice},'' \emph{IEEE J. on
  selected areas in commun.}, vol.~35, no.~6, pp. 1201--1221, Apr. 2017.

\bibitem{Intro4}
P.~Pirinen, ``{A brief overview of 5G research activities},'' in \emph{Proc.
  IEEE 5GU}, Akaslompolo, Finland, Nov. 2014.

\bibitem{Intro5}
A.~Al-Fuqaha, M.~Guizani, M.~Mohammadi, M.~Aledhari, and M.~Ayyash, ``Internet
  of things: A survey on enabling technologies, protocols, and applications,''
  \emph{IEEE Commun. Surveys \& Tutorials}, vol.~17, no.~4, pp. 2347--2376,
  Jun. 2015.

\bibitem{Intro6}
M.~R. Palattella \emph{et~al.}, ``{Internet of things in the 5G era: Enablers,
  architecture, and business models},'' \emph{{IEEE} J. Sel. Areas Commun.},
  vol.~34, no.~3, pp. 510--527, Feb. 2016.

\bibitem{Intro7}
D.~Tsonev, S.~Videv, and H.~Haas, ``{Towards a 100 Gb/s visible light wireless
  access network},'' \emph{Optics express}, vol.~23, no.~2, pp. 1627--1637,
  Jan. 2015.

\bibitem{Intro8}
N.~Bhushan, J.~Li, D.~Malladi, R.~Gilmore, D.~Brenner, A.~Damnjanovic,
  R.~Sukhavasi, C.~Patel, and S.~Geirhofer, ``{Network densification: the
  dominant theme for wireless evolution into 5G},'' \emph{IEEE Commun.
  Magazine}, vol.~52, no.~2, pp. 82--89, Feb. 2014.

\bibitem{Intro9}
X.~Ge, S.~Tu, G.~Mao, C.-X. Wang, and T.~Han, ``{5G ultra-dense cellular
  networks},'' \emph{IEEE Wireless Commun.}, vol.~23, no.~1, pp. 72--79, Mar.
  2016.

\bibitem{islam2017power}
S.~R. Islam \emph{et~al.}, ``{Power-domain non-orthogonal multiple access
  (NOMA) in 5G systems: Potentials and challenges},'' \emph{IEEE Commun.
  Surveys \& Tutorials}, vol.~19, no.~2, pp. 721--742, Oct. 2017.

\bibitem{6861434}
X.~{Zhang} and M.~{Haenggi}, ``The performance of successive interference
  cancellation in random wireless networks,'' \emph{{IEEE} Trans. Inform.
  Theory}, vol.~60, no.~10, pp. 6368--6388, Oct 2014.

\bibitem{liu2016capacity}
L.~Liu, C.~Yuen, Y.~L. Guan, and Y.~Li, ``{Capacity-achieving iterative LMMSE
  detection for MIMO-NOMA systems},'' in \emph{Proc. IEEE ICC}, Kuala Lumpur,
  Malaysia, May. 2016.

\bibitem{li2015investigation}
A.~Li \emph{et~al.}, ``{Investigation on hybrid automatic repeat request (HARQ)
  design for NOMA with SU-MIMO},'' in \emph{Proc. IEEE PIMRC}, Hong Kong,
  China, Sept. 2015.

\bibitem{zhao2016noma}
J.~Zhao, Y.~Liu, K.~K. Chai, Y.~Chen, M.~Elkashlan, and J.~Alonso-Zarate,
  ``{NOMA-based D2D communications: Towards 5G},'' in \emph{Proc. IEEE
  Globecom}, Washington, DC, USA, Dec. 2016.

\bibitem{Ding:NOMA:Application}
Z.~{Ding}, Y.~{Liu}, J.~{Choi}, Q.~{Sun}, M.~{Elkashlan}, C.~{I}, and H.~V.
  {Poor}, ``Application of non-orthogonal multiple access in {LTE} and {5G}
  networks,'' \emph{{IEEE} Commun. Mag.}, vol.~55, no.~2, pp. 185--191,
  February 2017.

\bibitem{islam2017powerNOMA}
S.~R. Islam \emph{et~al.}, ``{Power-domain non-orthogonal multiple access
  (NOMA) in 5G systems: Potentials and challenges},'' \emph{IEEE Commun.
  Surveys \& Tutorials}, vol.~19, no.~2, pp. 721--742, Oct. 2017.

\bibitem{PA:NOMA1}
X.~{Song}, L.~{Dong}, J.~{Wang}, L.~{Qin}, and X.~{Han}, ``Energy efficient
  power allocation for downlink {NOMA} heterogeneous networks with imperfect
  {CSI},'' \emph{{IEEE} Access}, vol.~7, pp. 39\,329--39\,340, 2019.

\bibitem{PA:NOMA2}
J.~{Chen}, J.~{Jia}, Y.~{Liu}, X.~{Wang}, and A.~H. {Aghvami}, ``Optimal
  resource block assignment and power allocation for {D2D}-enabled {NOMA}
  communication,'' \emph{{IEEE} Access}, vol.~7, pp. 90\,023--90\,035, 2019.

\bibitem{PA:NOMA3}
K.~{Wang}, Y.~{Liu}, Z.~{Ding}, A.~{Nallanathan}, and M.~{Peng}, ``User
  association and power allocation for multi-cell non-orthogonal multiple
  access networks,'' \emph{{IEEE} Trans. Wireless Commun.}, pp. 1--1, 2019.

\bibitem{PA:NOMA4}
M.~{Zeng}, A.~{Yadav}, O.~A. {Dobre}, and H.~V. {Poor}, ``Energy-efficient
  joint user-rb association and power allocation for uplink hybrid noma-oma,''
  \emph{{IEEE} Internet Things J.}, vol.~6, no.~3, pp. 5119--5131, June 2019.

\bibitem{Islam:ResourceAllocation:NOMA}
S.~M.~R. {Islam}, M.~{Zeng}, O.~A. {Dobre}, and K.~{Kwak}, ``Resource
  allocation for downlink {NOMA} systems: Key techniques and open issues,''
  \emph{{IEEE} Wireless Commun.}, vol.~25, no.~2, pp. 40--47, April 2018.

\bibitem{NOMAPairing:MatchingAlgorithm}
W.~{Liang} \emph{et~al.}, ``{User Pairing for Downlink Non-Orthogonal Multiple
  Access Networks Using Matching Algorithm},'' \emph{{IEEE} Trans. Commun.},
  vol.~65, no.~12, pp. 5319--5332, Aug. 2017.

\bibitem{Letter:NOMA:Pairing}
L.~{Zhu}, J.~{Zhang}, Z.~{Xiao}, X.~{Cao}, and D.~O. {Wu}, ``{Optimal User
  Pairing for Downlink Non-Orthogonal Multiple Access (NOMA)},'' \emph{{IEEE}
  Wireless Commun. Lett.}, vol.~8, no.~2, pp. 328--331, Apr. 2019.

\bibitem{Ding:Impact:Pairing}
Z.~{Ding}, P.~{Fan}, and H.~V. {Poor}, ``{Impact of user pairing on 5G non
  orthogonal multiple-access downlink transmissions},'' \emph{IEEE Trans. on
  Vehicular Technology}, vol.~65, no.~8, pp. 6010--6023, Sep. 2016.

\bibitem{Ding:NOMA:ModeSelection}
D.~{Zhai}, R.~{Zhang}, Y.~{Wang}, H.~{Sun}, L.~{Cai}, and Z.~{Ding}, ``Joint
  user pairing, mode selection, and power control for d2d-capable cellular
  networks enhanced by nonorthogonal multiple access,'' \emph{{IEEE} Internet
  Things J.}, vol.~6, no.~5, pp. 8919--8932, Oct 2019.

\bibitem{Zhu:OptimalPairing}
L.~{Zhu}, J.~{Zhang}, Z.~{Xiao}, X.~{Cao}, and D.~O. {Wu}, ``Optimal user
  pairing for downlink non-orthogonal multiple access (noma),'' \emph{IEEE
  Wireless Communications Letters}, vol.~8, no.~2, pp. 328--331, April 2019.

\bibitem{7117391}
Z.~Ding, M.~Peng, and H.~V. Poor, ``{Cooperative non-orthogonal multiple access
  in 5G systems},'' \emph{IEEE Commun. Letters}, vol.~19, no.~8, pp.
  1462--1465, Jun. 2015.

\bibitem{BoydS:98:LAA}
M.~Lobo, L.~Vandenberghe, S.~Boyd, and H.~Lebret, ``Applications of
  second-order cone programming,'' \emph{Lin. Alg. and its Applications}, vol.
  284, pp. 193--228, Jan. 1998.

\bibitem{C-NOMA:2users1}
F.~{Kara} and H.~{Kaya}, ``On the error performance of cooperative-{NOMA} with
  statistical {CSIT},'' \emph{{IEEE} Commun. Lett.}, vol.~23, no.~1, pp.
  128--131, Jan 2019.

\bibitem{C-NOMA:2users2}
Q.~Y. {Liau} and C.~Y. {Leow}, ``Successive user relaying in cooperative {NOMA}
  system,'' \emph{{IEEE} Wireless Commun. Lett.}, vol.~8, no.~3, pp. 921--924,
  June 2019.

\bibitem{Liu:C-NOMA:2users}
G.~{Liu}, X.~{Chen}, Z.~{Ding}, Z.~{Ma}, and F.~R. {Yu}, ``Hybrid
  half-duplex/full-duplex cooperative non-orthogonal multiple access with
  transmit power adaptation,'' \emph{IEEE Trans. on Wireless Commun.}, vol.~17,
  no.~1, pp. 506--519, Jan. 2018.

\bibitem{C-NOMA:2users3}
B.~{Ning}, W.~{Hao}, A.~{Zhang}, J.~{Zhang}, and G.~{Gui}, ``Energy
  efficiency–delay tradeoff for a cooperative {NOMA} system,'' \emph{{IEEE}
  Commun. Lett.}, vol.~23, no.~4, pp. 732--735, April 2019.

\bibitem{8417519}
T.~E.~A. {Alharbi} and D.~K.~C. {So}, ``Full-duplex decode-and-forward
  cooperative non-orthogonal multiple access,'' in \emph{{Proc. IEEE VTC}},
  Porto, Portugal, Jun. 2018.

\bibitem{Guo:MultiuserCNOMA}
N.~{Guo}, J.~{Ge}, Q.~{Bu}, and C.~{Zhang}, ``Multi-user cooperative
  non-orthogonal multiple access scheme with hybrid full/half-duplex
  user-assisted relaying,'' \emph{{IEEE} Access}, vol.~7, pp. 39\,207--39\,226,
  2019.

\bibitem{Zhou:multiuserCNOMA1}
Y.~{Zhou}, V.~W.~S. {Wong}, and R.~{Schober}, ``Dynamic decode-and-forward
  based cooperative noma with spatially random users,'' \emph{{IEEE} Trans.
  Wireless Commun.}, vol.~17, no.~5, pp. 3340--3356, May 2018.

\bibitem{Liu:MultiuserCNOMA2:SWIPT}
Y.~{Liu}, Z.~{Ding}, M.~{Elkashlan}, and H.~V. {Poor}, ``Cooperative
  non-orthogonal multiple access with simultaneous wireless information and
  power transfer,'' \emph{{IEEE} J. Sel. Areas Commun.}, vol.~34, no.~4, pp.
  938--953, April 2016.

\bibitem{Alouini:VLC:C-NOMA}
\BIBentryALTinterwordspacing
M.~Obeed \emph{et~al.}, ``User pairing, link selection and power allocation for
  cooperative {NOMA} hybrid {VLC/RF} systems,'' \emph{CoRR}, vol.
  abs/1908.10803, 2019. [Online]. Available:
  \url{http://arxiv.org/abs/1908.10803}
\BIBentrySTDinterwordspacing

\bibitem{KTJ:16:TSP}
J.~Kaleva, A.~Tölli, and M.~J. Juntti, ``Decentralized sum rate maximization
  with {QoS} constraints for interfering broadcast channel via successive
  convex approximation,'' \emph{{IEEE} Trans. Signal Processing}, to appear.

\bibitem{6204010}
S.~{Vanka}, S.~{Srinivasa}, Z.~{Gong}, P.~{Vizi}, K.~{Stamatiou}, and
  M.~{Haenggi}, ``Superposition coding strategies: Design and experimental
  evaluation,'' \emph{{IEEE} Trans. Wireless Commun.}, vol.~11, no.~7, pp.
  2628--2639, July 2012.

\bibitem{Caire:DecodeAndForward}
K.~R. {Kumar} and G.~{Caire}, ``Coding and decoding for the dynamic decode and
  forward relay protocol,'' \emph{{IEEE} Trans. Inform. Theory}, vol.~55,
  no.~7, pp. 3186--3205, July 2009.

\bibitem{Bossert:ChannelCodingBook}
M.~Bossert, \emph{Channel Coding for Telecommunications}, 1st~ed.\hskip 1em
  plus 0.5em minus 0.4em\relax New York, NY, USA: John Wiley \& Sons, Inc.,
  1999.

\bibitem{Laneman:CooperativeDiveristy}
J.~N. {Laneman}, D.~N.~C. {Tse}, and G.~W. {Wornell}, ``Cooperative diversity
  in wireless networks: Efficient protocols and outage behavior,'' \emph{{IEEE}
  Trans. Inform. Theory}, vol.~50, no.~12, pp. 3062--3080, Dec 2004.

\bibitem{8094966}
G.~{Liu}, X.~{Chen}, Z.~{Ding}, Z.~{Ma}, and F.~R. {Yu}, ``Hybrid
  half-duplex/full-duplex cooperative non-orthogonal multiple access with
  transmit power adaptation,'' \emph{IEEE Trans. on Wireless Commun.}, vol.~17,
  no.~1, pp. 506--519, Jan. 2018.

\bibitem{1310306}
S.~{Roy} \emph{et~al.}, ``Maximal-ratio combining architectures and performance
  with channel estimation based on a training sequence,'' \emph{{IEEE} Trans.
  Wireless Commun.}, vol.~3, no.~4, pp. 1154--1164, July 2004.

\bibitem{4205048}
J.~{Tang} and X.~{Zhang}, ``Cross-layer resource allocation over wireless relay
  networks for quality of service provisioning,'' \emph{{IEEE} J. Sel. Areas
  Commun.}, vol.~25, no.~4, pp. 645--656, May 2007.

\bibitem{Tseng:Convergence:BCD}
\BIBentryALTinterwordspacing
P.~Tseng, ``Convergence of a block coordinate descent method for
  nondifferentiable minimization,'' \emph{J. Optim. Theory Appl.}, vol. 109,
  no.~3, pp. 475--494, Jun. 2001. [Online]. Available:
  \url{http://dx.doi.org/10.1023/A:1017501703105}
\BIBentrySTDinterwordspacing

\bibitem{bard2013practical}
J.~F. Bard, \emph{Practical bilevel optimization: algorithms and
  applications}.\hskip 1em plus 0.5em minus 0.4em\relax Springer Science \&
  Business Media, 2013, vol.~30.

\bibitem{Colson2007AnOO}
B.~Colson, P.~Marcotte, and G.~Savard, ``An overview of bilevel optimization,''
  \emph{Annals of Operations Research}, vol. 153, pp. 235--256, 2007.

\bibitem{Li:ModeSelection:Letter}
Y.~{Li}, T.~{Wang}, Z.~{Zhao}, M.~{Peng}, and W.~{Wang}, ``Relay mode selection
  and power allocation for hybrid one-way/two-way half-duplex/full-duplex
  relaying,'' \emph{{IEEE} Commun. Lett.}, vol.~19, no.~7, pp. 1217--1220, July
  2015.

\bibitem{ebbesen2012generic}
S.~Ebbesen, P.~Kiwitz, and L.~Guzzella, ``A generic particle swarm optimization
  matlab function,'' in \emph{2012 American Control Conference (ACC)}.\hskip
  1em plus 0.5em minus 0.4em\relax IEEE, 2012, pp. 1519--1524.

\end{thebibliography}
\end{document}